\documentclass[runningheads]{llncs}
\usepackage{graphicx}
\usepackage{adjustbox}
\usepackage[inline]{enumitem}
\usepackage{blindtext}
\usepackage{orcidlink}
\usepackage{caption}
\usepackage{subcaption}
\usepackage{listings}
\usepackage{lstautogobble}
\usepackage{tikz}
\usepackage{amsmath}
\usepackage{amssymb}
\usepackage{stmaryrd}
\usepackage{cleveref}
\usepackage{todonotes}
\usepackage{wrapfig}
\usepackage{mathtools}
\usepackage{mleftright}
\usepackage{ifthen}
\usepackage{booktabs}
\usepackage{marvosym}

\newboolean{fullversion}
\setboolean{fullversion}{true}

\usetikzlibrary{calc,positioning,fit,backgrounds,matrix,shapes,decorations,decorations.pathreplacing,patterns}
\tikzset{shorten >=1pt,semithick}

\definecolor{bluekeywords}{rgb}{0.13, 0.13, 1}
\definecolor{greentypes}{rgb}{0, 0.5, 0}
\definecolor{orangecomments}{rgb}{1, 0.5, 0.1}
\definecolor{redstrings}{RGB}{171, 114, 2}
\definecolor{graynumbers}{rgb}{0.5, 0.5, 0.5}
\definecolor{goldcomments}{rgb}{0.6, 0.4, 0.08}

\definecolor{implemented}{rgb}{0.67, 0.9, 0.93}
\colorlet{existing}{lightgray}

\lstdefinelanguage{Lola}{
  keywords=[0]{input, output, trigger, constant, import, spawn, eval, close, with, when},
  moredelim=**[is][\transparent{0.6}]{?}{?},
  moredelim=**[is][\color{greentypes}@]{@}{@},
  keywordstyle=[0]\bfseries\color{bluekeywords},
  keywords=[1]{if, then, else, aggregate, defaults, offset, last, by, or, to, sin, cos, abs, hold, over, using, over_instances, prob},
  keywords=[2]{Variable, String, Int, Int64, UInt, UInt64, Bool, Float32, Float64, Float},
  keywordstyle=[2]\color{greentypes},
  sensitive=false,
  comment=[l]{//},
  morecomment=[s]{/*}{*/},
  morestring=[b]',
  morestring=[b]",
  literate={\\@}{@}1,
  moredelim=[s][\color{purple}\bfseries]{\#[}{]},
}
\makeatletter
\newcommand{\randFunc}{\mathcal{K}}

\newcommand{\Lap}{\mathit{Lap}}
\newcommand{\NN}{\mathbb{N}}
\newcommand{\BB}{\mathbb{B}}
\newcommand{\RR}{\mathbb{R}}
\newcommand{\VV}{\mathbb{V}}
\newcommand{\set}[1]{\left\{#1\right\}}
\newcommand{\setCond}[2]{\left\{#1\;\mid\;#2\right\}}
\newcommand{\floor}[1]{\left\lfloor#1\right\rfloor}

\newcommand{\pacingModel}{\mathbb{P}}
\newcommand{\timeMap}{\text{TimeMap}}
\newcommand{\pacingMap}{\text{PacingMap}}
\newcommand{\streamMap}{\text{StreamMap}}
\newcommand{\Stream}{\text{Stream}}
\renewcommand{\time}{\text{Time}}
\newcommand{\sr}{\mathtt{ID}}
\newcommand{\Iref}{\sr^\uparrow}
\newcommand{\Oref}{\sr^\downarrow}
\newcommand{\POref}{\sr^\downarrow_{\text{pub}}}
\newcommand{\World}{\mathbb{W}}
\newcommand{\world}{\omega}

\newcommand{\offset}{\mathsf{offset}}
\newcommand{\sync}{\mathsf{sync}}

\newcommand{\olast}{\mathit{last}}
\newcommand{\windowtimes}{\mathit{wtimes}}

\begin{document}
\title{Differentially Private Runtime Monitoring}
\titlerunning{Differentially Private Runtime Monitoring}
\author{
  Bernd Finkbeiner\orcidlink{0000-0002-4280-8441}
  \and Frederik Scheerer\textsuperscript{(\Letter)}\orcidlink{0009-0007-8115-0359}
}
\authorrunning{Finkbeiner and Scheerer}

\institute{CISPA Helmholtz Center for Information Security,\\Saarbrücken, Germany\newline
\email{\{finkbeiner, frederik.scheerer\}@cispa.de}}
\maketitle
\begin{abstract}
  Modern stream-based monitors collect detailed statistics of the runtime behavior of the system under observation.
  If the system runs in a privacy-sensitive context, this poses the risk of disclosing sensitive information.
  Differential privacy is the state-of-the-art approach for protecting sensitive information, however, integrating it into runtime monitoring is challenging: temporal operators can cause individual input values to influence multiple outputs over time, leading to repeated disclosure of private information.
  We propose an approach that automatically enforces differential privacy in stream-based monitoring specifications by analyzing temporal dependencies and injecting carefully calibrated noise into the specification.
  To preserve the utility of the outputs, we identify strategically chosen positions in the specification for noise injection and leverage tree-based mechanisms to mitigate the accuracy loss caused by noise injected into aggregation operators.
  We demonstrate the practicality and effectiveness of our approach in a case study on monitoring public transportation usage.

\end{abstract}

\section{Introduction}

Runtime monitoring is a powerful technique for ensuring a system's correctness during operation, especially in complex, real-world applications where static verification is difficult to obtain.
In many cases, a binary yes/no verdict is insufficient, and instead, stream-based languages output values that summarize and aggregate the system's current state over time.
At the same time, runtime monitoring is increasingly deployed in privacy-sensitive domains, where richer monitor outputs increase the risk of revealing private information.

For example, consider a car-rental company that continuously monitors the distance driven by each vehicle in order to schedule maintenance, detect abnormal wear, or ensure contractual conditions.
Even if the monitor reports only aggregated values, such as average daily distance driven, these outputs may still reveal sensitive information about the renter's behavior, including typical usage patterns and inferred routines when observed over time.
Similar concerns arise when monitoring commercial vehicle fleets, where runtime monitors are used to assess efficiency, safety, or regulatory compliance but inadvertently expose routes or operating schedules of drivers.
While runtime monitoring is essential for ensuring safety, performance, and compliance in such systems, it is crucial to ensure that the monitoring process does not become a source of privacy violations.

Differential privacy~(DP)~\cite{DBLP:conf/tcc/DworkMNS06} is widely regarded as the gold standard for providing formal privacy guarantees.
Informally, it ensures that the presence or absence of a single individual in the dataset has minimal impact on the output of a query.
This is typically achieved by adding carefully calibrated noise to the output and thereby producing a mechanism with provable privacy guarantees.

Integrating differential privacy into stream-based monitoring, however, introduces unique challenges due to the temporal nature of streams and the temporal operators describing them.
In contrast to static queries, monitors repeatedly release information over time, often derived from overlapping portions of the input data.
Releasing information derived from a single data point multiple times increases the overall privacy loss, as each release leaks additional information.
With temporal operators such as sliding-window aggregations, we observe exactly this behavior, and to reason about privacy in this setting, it is necessary to understand how individual inputs propagate through the temporal operators.

This observation introduces two key challenges.
First, we must be able to precisely characterize and compute the impact of individual inputs on the outputs, accounting for the effects of asynchronous streams and temporal operators, such as sliding-window aggregations.
Second, beyond providing formal privacy guarantees, we must ensure that the resulting private monitors remain accurate.
In particular, temporal operators amplify the influence of individual inputs, and the required amount of noise may degrade the quality of monitoring outputs.

In this paper, we present a runtime monitoring framework that enables users without privacy expertise to deploy monitors that prevent the leakage of sensitive information.
Our framework integrates differential privacy by analyzing the specification and automatically injecting appropriate amounts of noise.
It strategically selects positions in the specification that are beneficial to the accuracy of the output, and leverages tree-based constructs to further improve monitoring accuracy.
We demonstrate our approach using the stream-based specification language RTLola~\cite{DBLP:conf/fm/BaumeisterFKS24}, which has already been applied in several real-world, privacy-sensitive areas~\cite{DBLP:conf/tacas/BiewerFHKSS21,DBLP:conf/rv/FaymonvilleFST16,DBLP:conf/tacas/BaumeisterFSSW25,finkbeiner2021robust}.
In a case study on public transportation usage, we show that our approach enables private monitoring with high accuracy.

The remainder of this paper is organized as follows.
\Cref{sec:overview} gives an intuition of our approach using a motivating example, and \Cref{sec:dp_preliminaries} introduces our formalization of differential privacy for stream-based monitors.
In \Cref{sec:dp_lola}, we detail the sensitivity calculations for stream-based specifications, while \Cref{sec:accuracy} discusses our strategies for improving the accuracy of the private monitors.
Finally, \Cref{sec:stricter} explores extensions of our approach to stricter privacy guarantees, and \Cref{sec:evaluation} reports on experimental results.

\subsection{Related Work}

Closest to our work is the approach of Henzinger et al.~\cite{DBLP:conf/ccs/HenzingerKT25}, which studies private runtime monitoring based on private function evaluation while maintaining an internal state.
Further, prior work has explored alternative notions of privacy in LTL monitoring~\cite{DBLP:conf/emsoft/Abbas19,banno2022oblivious}.
However, these approaches do not rely on differential privacy and do not support the rich output provided by stream-based languages.

We are the first to integrate DP directly into runtime verification frameworks.
Previous work has focused on enforcing DP on database-style queries, especially SQL-like languages~\cite{DBLP:journals/pvldb/JohnsonNS18,DBLP:conf/innovations/BlockiBDS13,DBLP:journals/cacm/McSherry10,DBLP:conf/sigmod/MohanTSSC12,DBLP:conf/osdi/NarayanH12,DBLP:conf/stoc/NissimRS07,DBLP:journals/pvldb/ProserpioGM14}.
Some related work also considers queries on data streams~\cite{DBLP:conf/ipccc/LiWWWZM24}.
All these approaches, however, are designed for database-style queries and emphasize operators such as joins, rather than the temporal operators used in runtime verification.

An algorithmic idea adopted in our work is tree-based mechanisms for continual queries under DP.
Dwork et al.~\cite{DBLP:conf/stoc/DworkNPR10} introduced tree-based mechanisms for continual counting queries and Chan et al.~\cite{DBLP:journals/tissec/ChanSS11} independently proposed a similar approach for infinite streams.
Tree-based mechanisms were later extended to other aggregations such as sums~\cite{DBLP:conf/ndss/PerrierAK19,DBLP:conf/ccs/0001C0SC0LJ21,DBLP:journals/pvldb/QardajiYL13,DBLP:conf/aistats/Cardoso022}.
In addition to tree-based approaches, matrix-based mechanisms have been proposed for streaming data~\cite{DBLP:conf/focs/DvijothamMP0T24,DBLP:conf/ipccc/LiWWWZM24,DBLP:conf/soda/HenzingerUU23,DBLP:conf/soda/HenzingerUU24}, offering alternative trade-offs between accuracy and efficiency.
Some approaches also enforce DP for sliding window aggregations~\cite{DBLP:conf/edbt/CaoXGLBT13,DBLP:conf/icml/Upadhyay19,DBLP:journals/fcsc/ChenNZX23}, a temporal operator commonly used in stream-based monitoring.
Most streaming DP works assume event-level privacy and independent data; however, some approaches support user-level privacy~\cite{DBLP:conf/sp/FengMWLQH24,DBLP:conf/sp/DongLY23,DBLP:conf/ccs/ChenMHM17,DBLP:conf/icde/CaoYX017} or correlated data~\cite{DBLP:conf/cikm/FanX12,DBLP:conf/sigmod/SongWC17}.

We demonstrate our approach using the stream-based specification language RTLola~\cite{DBLP:conf/fm/BaumeisterFKS24}, an extension of Lola~\cite{DBLP:conf/time/DAngeloSSRFSMM05} that incorporates asynchronous streams and real-time functionality.
RTLola has been applied in multiple domains, including several privacy-sensitive domains such as monitoring sensor data from private vehicles~\cite{DBLP:conf/tacas/BiewerFHKSS21}, analyzing network traffic~\cite{DBLP:conf/rv/FaymonvilleFST16}, monitoring automatic decision and prediction systems~\cite{DBLP:conf/tacas/BaumeisterFSSW25}, and processing medical health data~\cite{finkbeiner2021robust}.
Although our approach is presented in the context of RTLola, the underlying ideas extend to other stream-based specification languages such as TeSSLa~\cite{kallwies2022tessla} and Striver~\cite{gorostiaga2018striver}.
\section{Motivating Example}\label{sec:overview}

We begin with a motivating example that illustrates the key ideas of our approach.
The methods will be detailed in the following sections.
We work in the setting of stream-based monitoring, where specifications define sets of streams, representing infinite sequences of values.
Input streams represent incoming data from the monitored system, while output streams are defined through stream equations and compute new streams by transforming and aggregating the inputs.

\begin{figure}[t]
	\begin{lstlisting}
		input score : Int64	// $\mathcolor{orange}{\in \set{1,2,3,4,5,6}}$
		input conf : Int64 	// $\mathcolor{orange}{\in \set{-1,0,1}}$
		output adj := (6-score)*3 + conf + 1
		output davg @1d@ :=
			adj.aggregate(over: 3d, using: avg).defaults(to: 0.0)
		output low @1d@ := min(low.offset(by: -1).defaults(to: 15.0), davg)
		output high @1d@ := max(high.offset(by: -1).defaults(to: 0.0), davg)
		#[public]
		output range @1d@ := (low, high)
	\end{lstlisting}
	\caption{An example RTLola specification calculating statistics of user feedback.}
	\label{fig:example_rtlola}
\end{figure}

Consider \Cref{fig:example_rtlola}, which shows an example for the stream-based specification language RTLola~\cite{DBLP:conf/fm/BaumeisterFKS24}.
It models a monitoring scenario in which user feedback is continually collected and evaluated over time.
Each user's feedback consists of new values added to the two input streams:
a \lstinline!score!, ranging from 1 (best) to 6 (worst), and a \lstinline!conf! value indicating positive (1), neutral (0), or negative (-1) confidence.
The specification combines these inputs into a single adjusted rating between 0 and 17, and emits the result as a new value of the output stream \lstinline!adj!.
Using a sliding window aggregation, the stream \lstinline!davg! computes the average rating over the last three days.
Finally, the streams \lstinline!low! and \lstinline!high! track the minimum and maximum daily averages observed so far.
The \lstinline!range! stream, consisting of both statistics, is marked as public using RTLola's annotation mechanism~\cite{DBLP:conf/rv/BaumeisterFS25}, indicating that its values are output by the monitor to the public.

Individual user feedback is privacy-sensitive.
While we assume the monitor itself is trusted, its public outputs go into an untrusted world.
Therefore, the monitor must ensure that no private information is leaked by observing the publicly emitted values of \lstinline!range! over time.
To incorporate differential privacy into a stream-based specification, we must understand how a change to a single input influences observable outputs.
Offset accesses propagate information across time, as seen in the recursive definitions of \lstinline!low! and \lstinline!high!, where a single rating influences arbitrarily many future outputs.
Similarly, in sliding-window aggregations, each input contributes to multiple overlapping windows, so a single change affects multiple results.

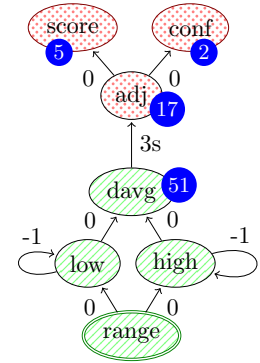
\begin{wrapfigure}[17]{r}{0.3\linewidth}
	\centering
	\scalebox{0.9}{
	\begin{tikzpicture}[n/.style={draw,ellipse,minimum height=7mm,inner sep=0.8mm},sens/.style={color=white,fill=blue,circle,font=\footnotesize,inner sep=0.5mm},ch/.style={pattern=crosshatch dots,pattern color=red!50!white},li/.style={pattern=north east lines, pattern color=green!50!white}]
		\node[n,ch,draw=red!50!black] (mark) {score};
		\node[n,right=5mm of mark,ch,draw=red!50!black] (confidence) {conf};
		\node[n,ch] (adj) at ($(mark)!0.5!(confidence) + (0,-10mm)$) {adj};
		\node[n,below=7mm of adj,li] (avg) {davg};
		\node[n,below left=8mm and 2mm of avg,anchor=center,li] (low) {low};
		\node[n,below right=8mm and 2mm of avg,anchor=center,li] (high) {high};
		\node[n,below=14mm of avg,li,double,draw=green!50!black] (range) {range};

		\draw[->] (adj) -- node[below left] {0} (mark);
		\draw[->] (adj) -- node[below right] {0} (confidence);
		\draw[->] (avg) -- node[right] {3s} (adj);
		\draw[->] (low) -- node[above left,pos=0.2] {0} (avg);
		\draw[->] (high) -- node[above right,pos=0.2] {0} (avg);
		\draw[->] (range) -- node[below left,pos=0.8] {0} (low);
		\draw[->] (range) -- node[below right,pos=0.8] {0} (high);
		\draw[->] (low) to[loop left] node[above,pos=0.8] {-1} (low);
		\draw[->] (high) to[loop right] node[above,pos=0.2] {-1} (high);
		\node[below=-2mm of mark,sens,xshift=-2mm] {5};
		\node[below=-2mm of confidence,sens,xshift=2mm] {2};
		\node[right=-2mm of adj,sens,yshift=-2mm] {17};
		\node[right=-2mm of avg,sens,yshift=1mm] {51};
	\end{tikzpicture}
	}
	\caption{The dependency graph for the specification in \Cref{fig:example_rtlola}.}
	\label{fig:dep_graph_example}
\end{wrapfigure}
Dependencies are made explicit in the \emph{dependency graph}~\cite{DBLP:conf/fm/BaumeisterFKS24} of the specification.
Each stream in the specification corresponds to a node, and stream accesses are represented by directed edges in the graph, annotated with the access kind.
The dependency graph for the example is shown in \Cref{fig:dep_graph_example}.
In the graph, the private input streams are marked with a red outline, but streams that (indirectly) depend on private inputs also contain private information.
The extent to which private information influences a stream is quantified by the \emph{sensitivity}.
Intuitively, the sensitivity of a stream bounds how much its values can change when a single user's feedback is modified.
In a temporal setting, this captures both how strongly values change and how often this change reappears over time.
Sensitivity propagates through the dependency graph, accumulating at arithmetic operations and amplifying across temporal operators.

For instance, the sensitivity of the \lstinline!score! stream is 5, reflecting its input range.
The \lstinline|adj| stream combines multiple pieces of private information and therefore has higher sensitivity.
The sliding window further amplifies this:
A single value contributes to three overlapping windows, resulting in a sensitivity of 51 for \lstinline!davg!.
This sensitivity information allows us to track how private information flows through the specification.
Based on this information, we select a subset of nodes at which the monitor adds noise to each value.
Adding noise to a node renders the corresponding information in the stream privacy-preserving, which in turn makes all downstream computations safe to publish.
This yields a placement problem: where should noise be added so that all public outputs are protected, while keeping the loss of accuracy as small as possible?

In the example, we select the \lstinline!davg! stream and add noise calibrated to its sensitivity.
This choice ensures that all values flowing into the recursive definitions of \lstinline!low! and \lstinline!high! are already privacy-preserving.
Applying noise earlier, for example at the inputs, would reduce accuracy because each aggregated value would be noisy by itself.
After the aggregation, a single rating can be hidden among multiple others, making its individual contribution smaller and requiring less perturbation relative to the result.
Now, as this node is privacy-preserving, it and all downstream nodes are marked as green.
As all public outputs are among these green nodes, the overall specification satisfies the desired privacy guarantees.

In this example, the sliding window amplifies the sensitivity linearly with the window size.
Later in the paper, we show how tree-based aggregation methods can substantially increase accuracy for longer windows.
\section{Differential Privacy of Stream-Based Monitors}\label{sec:dp_preliminaries}

This section formalizes differential privacy for stream-based runtime monitoring.
We first summarize the standard definitions and mechanisms of differential privacy, stated with respect to a general adjacency relation.
We then present the evaluation model used to represent monitor executions and define an event-level adjacency notion capturing the privacy requirements for stream-based monitors.

\subsection{Differential Privacy Background}

Differential privacy~\cite{DBLP:conf/tcc/DworkMNS06} is the de facto standard for ensuring privacy, grounded in a rigorous mathematical definition.
Originally introduced in the context of databases, it provides a strong guarantee: the presence or absence of any individual in the dataset has only a minimal impact on the output of a query on the database.
Intuitively, a randomized mechanism is differentially private if an observer cannot distinguish whether a particular record was part of the database based on the output of the mechanism.
\begin{definition}[Differential Privacy~\cite{DBLP:journals/fttcs/DworkR14}]
	Let $\mathcal{D}$ be a domain and $\sim\;\subseteq \mathcal{D}\times\mathcal{D}$ a symmetric adjacency relation.
	A randomized mechanism $\randFunc: \mathcal{D} \rightarrow \mathcal{R}$ is $\varepsilon$-differentially private w.r.t. $\sim$ if for all $D,D'\in\mathcal{D}$ with $D \sim D'$ and for all $S \subseteq \mathcal{R}$,
	\[
		\Pr[\randFunc(D) \in S] \le \exp(\varepsilon) \cdot \Pr[\randFunc(D') \in S],
	\]
	where the probability is over the randomness of $\randFunc$.
\end{definition}
We state the definition in its general form with respect to an arbitrary adjacency relation.
The \emph{adjacency relation} $\sim$ specifies which pairs of inputs should be considered indistinguishable by an observer.
In the classical database setting, the relation captures neighboring databases that differ by a single row.
The Probability $\Pr[\randFunc(D) \in S]$ captures the chance that the mechanism $\randFunc$ produces an output in the set $S$.
Differential privacy bounds how much this probability can change between adjacent inputs.
The \emph{privacy parameter} $\varepsilon$ quantifies the strength of the privacy guarantee.
Smaller values of $\varepsilon$ correspond to stronger privacy.
Choosing $\varepsilon$ involves balancing the level of privacy with the output utility, depending on the privacy requirements of the data and the intended application.

To design differentially private mechanisms, we must understand how much the output of a query can change between two adjacent inputs.
This quantity is called the sensitivity of the query.
\begin{definition}[Sensitivity~\cite{DBLP:journals/fttcs/DworkR14}]
	For $f : \mathcal{D} \rightarrow \mathbf{R}^d$, the $L_1$ sensitivity of $f$ w.r.t. $\sim$ is $\Delta_\sim f = \max_{\substack{D,D'\\D\sim D'}} \|f(D) - f(D')\|_1$.
\end{definition}
Intuitively, the sensitivity measures the maximum possible change in the output between two adjacent inputs.
For instance, a query that counts the number of rows satisfying a certain predicate has sensitivity 1, since modifying or removing a single row can change the count by at most one.
In contrast, a query that sums numerical values can have a much larger sensitivity, depending on the possible ranges of those values.

Once the sensitivity of a query is known, differential privacy can be achieved by adding Laplace noise calibrated to that sensitivity.
\begin{theorem}[Laplace Mechanism~\cite{DBLP:journals/fttcs/DworkR14}]\label{thm:laplace}
	For $f : \mathcal{D} \rightarrow \mathbf{R}^d$, the mechanism $\randFunc(\mathcal{D}) = f(\mathcal{D}) + \eta$ that adds independently generated noise $\eta$ with distribution $\eta_i \sim \Lap(\Delta_\sim f / \varepsilon)$ to each of the $d$ output terms satisfies $\varepsilon$-differential privacy w.r.t. $\sim$.
\end{theorem}

While the Laplace mechanism establishes privacy guarantees for individual mechanisms, composition theorems describe how these guarantees change when multiple mechanisms are combined.
Under \emph{sequential composition}~\cite{DBLP:journals/fttcs/DworkR14}, publishing the outputs of multiple mechanisms results in an additive privacy loss: combining mechanisms with privacy parameters $\varepsilon_1$ and $\varepsilon_2$ gives an overall privacy guarantee of $(\varepsilon_1 + \varepsilon_2)$ with respect to the same adjacency relation.
As a result, repeatedly releasing noisy versions of the same query gradually weakens privacy, so individual mechanisms must use smaller privacy parameters to stay within a fixed overall budget.
\emph{Post-processing}~\cite{DBLP:journals/fttcs/DworkR14} ensures that any transformation applied to the output of a differentially private mechanism preserves its privacy guarantee.
Once a mechanism has produced a privacy-preserving output, it can be safely transformed without adding any extra privacy loss.
Together, these properties let us reason modularly about privacy in stream-based monitors.

\subsection{Evaluation Models for Runtime Monitoring}

To formalize differential privacy in the context of runtime monitoring, we need a model allowing us to compare different monitor executions.
For this purpose, we introduce the notion of an evaluation model.
Let $\time = \NN$ be the set of discrete timestamps, $\VV_\bot = \RR \cup \{\bot\}$ the set of optional stream values, and $\RR^+$ the set of non-negative real numbers.
\begin{definition}[Evaluation Model~\cite{DBLP:conf/fm/BaumeisterFKS24}]
	An \emph{evaluation model} $\world \in \World$ consists of timed data streams over a set of input stream references $\Iref$ and output stream references $\Oref$.
	It is represented as the combination of a stream map and a timing map:
	\begin{align*}
		\Stream &:= \time \rightarrow \VV_\bot\\
		\streamMap &:= \left(\Iref \uplus \Oref\right) \rightarrow \Stream\\
		\timeMap &:= \time \rightarrow \RR^+\\
		\World &:= \streamMap \times \timeMap.
	\end{align*}
\end{definition}
Intuitively, an evaluation model captures all values produced by the monitor during one execution.
The stream mapping assigns each input and output stream reference to a stream of values indexed by discrete timestamps.
Because the monitor operates asynchronously, a stream may have no new value at a given timestamp.
This absence is represented by the symbol $\bot$.
The time mapping associates each discrete timestamp with its corresponding real-time timestamp.

For convenience, given an evaluation model $\world = (\mathit{streams}, \mathit{times})$, we define $\world(t) := \mathit{times}(t)$ to denote the real-time timestamp of $t \in \time$, and $\world(s) := \mathit{streams}(s)$ to denote the stream assigned to reference $s \in \Iref \uplus \Oref$.

\subsection{Adjacency of Evaluation Models}

To reason about differential privacy in the context of monitor executions, we define a suitable adjacency relation between evaluation models.
Different choices of adjacency correspond to different privacy guarantees and reflect what aspect of the monitored behavior is intended to be protected.
In this work, we adopt an \emph{event-level}~\cite{DBLP:conf/stoc/DworkNPR10} notion of privacy:
Two evaluation models are adjacent if their input streams differ by exactly one event occurring at a single timestamp.
Moreover, we assume that only the content of an event is sensitive, not its timing.
Consequently, adjacent evaluation models share the same timing map, ensuring that temporal information remains unchanged:
\begin{definition}[Event-Level Adjacency]
	Two evaluation models $\world, \world' \in \World$ are considered \emph{event-level adjacent}, if
	\begin{enumerate}
		\item For all $t \in \time$, it holds that $\world(t) = \world'(t)$,
        \item For all $t \in \time$ and $i \in \Iref$, it holds that $\world(i)(t) = \bot \Leftrightarrow \world'(i)(t) = \bot$,
		\item There exists at most one $t \in \time$, such that $\exists i \in \Iref. \world(i)(t) \ne \world'(i)(t)$.
	\end{enumerate}
	Further, we consider them \emph{$s$-distant} if
	$
		\forall t \in \time. \forall i \in \Iref. |\world(i)(t) - \world'(i)(t)| \le s.
	$
\end{definition}
The notion of $s$-distance restricts how much two adjacent evaluation models may differ in their input values: each input can change by at most $s$.
This assumption reflects many realistic scenarios, where input domains are naturally bounded, for instance, sensor readings limited by the hardware, or demographic data such as age constrained within feasible ranges.

Our notion of privacy protects individual datapoints.
For example, it can hide a single detour when tracking a person's car, but it cannot hide the daily commute to work.
This choice is motivated by our approach using static analysis, and stronger notions of adjacency would require significant amounts of noise, making it difficult to produce accurate monitors.
In \Cref{sec:stricter}, we nevertheless explore strategies for supporting stronger notions of privacy within our framework.
\section{Sensitivity Calculations for Stream-Based Languages}\label{sec:dp_lola}

The sensitivity of a function measures how much its output can change for two adjacent inputs.
For stream-based monitors, we analyze how much output streams may differ between two adjacent and valid evaluation models.
We introduce a formal framework for the sensitivity of streams, discuss practical challenges in stream-based computations, and establish static bounds that enable differentially private monitoring.

\subsection{Specifications}
A \emph{specification} $\varphi$ consists of a set of stream equations $\{x := \varphi_x \mid x \in \Oref\}$, where each $\varphi_x$ is a stream expression over $\Iref \uplus \Oref$.
Each output stream is additionally annotated with a \emph{pacing type} that determines when it is evaluated relative to the input streams.
An evaluation model $\world \in \World$ is \emph{valid} for $\varphi$ if every output stream value $\world(x)(t)$ is computed according to the defining equation $\varphi_x$ for all $t \in \time$.

\subsection{Stream Expressions}

Sensitivities are defined recursively over stream expressions.
Expressions may reference a stream $y$, a constant $c$, an offset $o\in\NN$, or be combined using a binary operator~$\circ \in \set{+,-,\,\cdot\,}$:
\begin{align*}
\mathit{exp} ::= \quad&\mathsf{sync}(y) \mid \mathsf{offset}(y,o) \mid \mathsf{hold}(y) \mid \mathsf{aggr}(y,W,f)\\
&\mid \mathsf{default}(\mathit{exp},\mathit{exp}) \mid \mathsf{ite}(\mathit{exp}, \mathit{exp}, \mathit{exp}) \mid \mathit{exp} \circ \mathit{exp} \mid c
\end{align*}
We distinguish between \emph{synchronous accesses}, namely $\mathsf{sync}$ and $\mathsf{offset}$, which access the value of $y$ at the current (resp. offset) discrete timepoint, and \emph{asynchronous accesses}, such as $\mathsf{hold}$, which returns the most recent available value, and $\mathsf{aggr}$, which applies the aggregation function $f\in\set{\Sigma,\mathit{last},\mathit{count}}$ over a window $W$ of values of $y$.
As asynchronous accesses may be undefined at certain time points, we provide the $\mathsf{default}$ operator, which supplies a fallback value whenever an access fails.
\ifthenelse{\boolean{fullversion}}{
    A formal definition of the semantics of stream expressions is given in \Cref{app:semantics}.
}{
    For a formal definition of the semantics of stream expressions, we refer to the full version of this paper~\cite{fullversion}.
}

\subsection{Timing Abstraction}

Since the sensitivity analysis should be independent of the specific stream values but may depend on the timing of events, we separate the timing information of an evaluation model.
For this purpose, we introduce a \emph{pacing model}, which captures only the temporal structure of stream evaluations.
\begin{definition}[Pacing Model]
	A \emph{pacing model} $p \in \pacingModel$ captures the timing of stream evaluations independent of stream values.
	It consists of a pacing and timing map, where $\BB = \set{\bot,\top}$ denotes the set of boolean values:
	\begin{align*}
		\pacingMap &:= \left(\Iref \uplus \Oref\right) \rightarrow \time \rightarrow \BB\\
		\timeMap &:= \time \rightarrow \RR^+\\
		\pacingModel &:= \pacingMap \times \timeMap
	\end{align*}	
\end{definition}
Analogously to the evaluation model, we introduce shorthand notation.
For a pacing model $p = (\mathit{pacings}, \mathit{times})$, we define $p(t) := \mathit{times}(t)$ for a time $t \in \time$, and $p(s) := \mathit{pacings}(s)$ for a stream $s \in \Iref \uplus \Oref$.
Further, we say an evaluation model \emph{behaves according to} a pacing model if all real-time timestamps are equal and a stream has a value exactly if the corresponding pacing is true.

To enable an automatic analysis of specifications, we impose assumptions on the pacing of stream evaluations.
Among other checks, such as dependency and value type analysis, each specification is statically verified to satisfy these pacing constraints before execution.
These constraints guarantee that the monitoring algorithm produces only valid evaluation models and that the resulting monitor operates within bounded memory requirements.
\begin{definition}[Correctly Typed Pacing Model]\label{def:valid_pacing}
	For a pacing model $p \in \pacingModel$ to be considered correctly typed according to a specification $\varphi$, it must hold that	
	\begin{enumerate}
		\item For each synchronous access $\mathsf{sync}$ and $\mathsf{offset}$ from output stream $x \in \Oref$ to stream $y \in \Iref \uplus \Oref$, it must hold that whenever $x$ evaluates, $y$ evaluates as well: $\forall t \in \time. p(x)(t) \Rightarrow p(y)(t)$.
		\item For each sliding window aggregation $\mathsf{aggr}$ from output stream $x \in \Oref$ over a stream $y \in \Iref \uplus \Oref$, it must hold that $x$ is periodic with some period $\delta_x$, i.e. $\forall t \in \time. p(x)(t) \Rightarrow \exists k \in \NN. p(t) = k \cdot \delta_x$.
	\end{enumerate}
\end{definition}

A specification $\varphi$ is \emph{well-typed} if it statically satisfies the pacing constraints above.
For every valid evaluation model $\world \in \World$ of a well-typed specification, there exists a correctly-typed pacing model $p \in \pacingModel$ that $\world$ behaves according to.

\subsection{Per-Event Sensitivity}

To measure how much the values of an output stream can differ between two adjacent evaluation models, we introduce \emph{per-event sensitivity}, which captures the maximum impact a single changed input has on a stream value at a given timestamp, taking into account the timing constraints imposed by a pacing model.
\begin{definition}[Per-Event Sensitivity]\label{def:delta}
	Given a specification $\varphi$, a pacing model $p \in \pacingModel$, and two timestamps $t,t' \in \time$, the per-event sensitivity defines how much the value of a stream $x \in \Iref \uplus \Oref$ at timestamp $t'$ can change, if the inputs at timestamp $t$ are changed by at most $s$.
	\begin{align*}
		\Delta_{t,p,s}^{t'}&: \Iref \uplus \Oref \rightarrow \RR^+\\
		\Delta_{t,p,s}^{t'}(x) &= \begin{cases}
			s & \text{if } x \in \Iref \land t' = t\\
			0 & \text{if } x \in \Iref \land t' \ne t\\
			\Delta_{t,p,s}^{t'}(\varphi_x) & \text{if } x \in \Oref
		\end{cases}
	\end{align*}
	where $\Delta_{t,p,s}^{t'}(\varphi_x)$ computes the sensitivity for the stream expression $\varphi_x$ of $x$.
	Given stream expressions $e$, $e_1$ and $e_2$, the constant $c$, and the stream reference $y$, we define
	\begin{align*}
		\Delta_{t,p,s}^{t'}(e_1 + e_2) &= \Delta_{t,p,s}^{t'}(e_1) + \Delta_{t,p,s}^{t'}(e_2)\\
		\Delta_{t,p,s}^{t'}(c \cdot e) &= |c| \cdot \Delta_{t,p,s}^{t'}(e)\\
		\Delta_{t,p,s}^{t'}(\mathsf{sync}(y)) &= \Delta_{t,p,s}^{t'}(y)\\
		\Delta_{t,p,s}^{t'}(\mathsf{offset}(y, o)) &= \Delta_{t,p,s}^{\olast(y, t', p, o)}(y)\\
		\Delta_{t,p,s}^{t'}(\mathsf{aggr}(y, W, \Sigma)) &= \sum_{t'' \in \windowtimes(y, t', p, W)} \Delta_{t,p,s}^{t''}(y)
	\end{align*}
	where $\olast(y, t, p, o)$ returns the timestamp of the $o$-th most recent value of $y$, and $\windowtimes(y, t, p, W)$ returns timestamps of $y$ from the $W$ sized window behind $t$.
    \ifthenelse{\boolean{fullversion}}{
        Consider \Cref{app:delta} for the full definition with all operators.
    }{
        Consider the full version of this paper~\cite{fullversion} for the full definition with all operators.
    }
\end{definition}
For an input stream, the sensitivity depends on the timestamp $t'$ being considered.
If $t' = t$, the sensitivity is exactly $s$, reflecting the allowed change at that time.
For all other timestamps, the input sensitivity is zero, since only a single input value may be modified.
For an output stream $x$, the sensitivity is determined by its defining stream expression $\varphi_x$ in the specification.
Most operations compute values at the current timestamp, so recursion does not shift the timestamp.
An exception is the \textit{offset}-operator:
The sensitivity at timestamp $t'$ depends on the sensitivity of the value retrieved from the past.
As this is determined by the pacing model, the function $\olast$ identifies the timestamp of the accessed value.
In the case of summation over sliding windows, the sensitivity at time $t'$ accumulates contributions from all aggregated values within the window.

Our intuition of the per-event sensitivity function is captured by this theorem:
\begin{theorem}\label{thm:delta}
	Let $\world, \world' \in \World$ be valid evaluation models for an acyclic specification $\varphi$ which are s-distant, event-level adjacent, and differ at timestamp $t$.
	Then, for each stream $x \in \Iref \uplus \Oref$ and each timestamp $t' \in \time$,
	\[
		|\world(x)(t') -  \world'(x)(t')| \le \Delta_{t,p,s}^{t'}(x),
	\]
	where both evaluation models define a unique and common pacing model $p \in \pacingModel$.
\end{theorem}
The proof is by structural induction over the defining stream expressions and can be found
\ifthenelse{\boolean{fullversion}}{%
in~\Cref{app:proof_delta}.
}{%
in the full version of this paper~\cite{fullversion}.
}
Note that throughout this chapter, the dependency graph is assumed acyclic.
Dependency cycles are handled in \Cref{sec:segments}.

\subsection{Value-Dependent Operators}\label{sec:value_dependent}

So far, we have assumed that an expression's sensitivity depends only on the sensitivities of the accessed streams.
For some operators, such as multiplication and conditionals, the sensitivity also depends on the actual values.
For instance, in a conditional $\mathsf{ite}(c, e_1, e_2)$, the output may switch between $e_1$ and $e_2$ for neighboring executions, and the resulting sensitivity depends on the maximum difference between the two branches.
Similarly, for a multiplication, the output sensitivity scales with the operand values.
To handle such operators, we extend the analysis to propagate \emph{value ranges} alongside sensitivities.
Each expression is associated with an interval that over-approximates the values it may produce at runtime.
Using this range information, the per-event sensitivity function $\Delta_{t,p,s}^{t'}$ can be extended as follows:
\begin{definition}[Extended Per-Event Sensitivity]  
Let $\Delta_{t,p,s}^{t'}(x)$ be defined as in \Cref{def:delta} for all standard operators.
For operators whose sensitivity depends on the input values, we extend it by
\[
\Delta_{t,p,s}^{t'}(x) := 
\begin{cases}
ub(x) - lb(x) & \text{if $t \rightsquigarrow_p^x t'$} \\
0 & \text{otherwise,}
\end{cases}
\]
where $t \rightsquigarrow_p^x t'$ indicates that a change of inputs at time $t$ can influence the value of stream $x$ at time $t'$ under the pacing model $p$, and $ub(\cdot)$ and $lb(\cdot)$ are the propagated value bounds. Find the definitions
\ifthenelse{\boolean{fullversion}}{%
in \Cref{app:def_bounds}.
}{%
in the full version of this paper~\cite{fullversion}.
}
\end{definition}

Note that sliding-window aggregations do not preserve bounded value ranges, as their output range grows unboundedly even when the output sensitivity is bounded.
As a consequence, we impose an additional typing constraint: value-dependent operators are only permitted on expressions for which a finite value range can be statically derived.
Note that bounded value ranges can always be re-established explicitly using the $\mathsf{clamp}$ operator, which enforces fixed lower and upper bounds on the produced values.

\subsection{Asynchronous Accesses}

Stream-based specifications allow asynchronous stream evaluations, so a single input value may be accessed multiple times by dependent streams.
For some asynchronous operators, such as sliding-window aggregations, well-typed pacing restrictions ensure that the number of contributions of a single input is bounded.
However, the \emph{hold access} operator retrieves the most recent value of another stream, allowing a single input to influence an unbounded number of timestamps.

Consider a stream $c$ that accesses the latest value of $a$ whenever $b$ occurs:
\begin{lstlisting}
output c @b@ := a.hold().defaults(to: 0)
\end{lstlisting}
Because the timing of $a$ and $b$ need not be known in advance, it is impossible to determine how often a single value of $a$ is accessed, and thus how much it may contribute to $c$.

In general, unrestricted hold accesses prevent bounding the total effect of a single input across all timestamps.
However, in practice, these accesses can often be constrained either by bounding the time span during which a value is relevant or by limiting the number of accesses. 
For instance, a bounded hold can be expressed as a sliding-window aggregation, e.g., 
$\mathsf{aggr}(a, 5s, \mathit{last})$
or by explicitly restricting the number of accesses, e.g., 
$\mathsf{hold}_5(a)$ with a newly introduced hold variant.
Once such bounds are made explicit, the contribution of a single input value becomes finite, and a static global sensitivity bound can be derived.

\subsection{Static Bounds}

While the per-event sensitivity $\Delta$ provides a way to quantify the sensitivity at a specific timestamp, it depends on a concrete pacing model 
and assumes knowledge of the timestamp at which the inputs differ, which is both unavailable at runtime.
To address this, we overapproximate the total $L_1$-sensitivity of $\Delta$ accumulated across all timestamps, by deriving a static upper bound for each stream.
The only assumption we make is that the pacing model is well-typed, a property statically verified on the specification.
Under this assumption, the following result establishes the existence of a finite bound for every stream.
\begin{theorem}[Sensitivity Bound]\label{thm:bound}
	Given a well-typed specification $\varphi$ with an acyclic dependency graph without unbounded hold operators and a distance value $s$, it holds that
	\[
		\forall x \in \left(\Iref \uplus \Oref\right). \exists b_{\varphi,x,s}. \forall p \in \pacingModel_\varphi . \forall t \in \time. \left(\sum_{t' \in T_{p,x}} \Delta_{t,p,s}^{t'}(x)\right) \le b_{\varphi,x,s}.
	\]
	where $\pacingModel_\varphi$ denotes all correctly-typed pacing models according to $\varphi$ and $T_{p,x} = \setCond{t' \in \time}{p(x)(t')}$ denotes the timestamps at which stream $x$ is evaluated.
\end{theorem}
Intuitively, this theorem ensures that each stream $x$ has a finite, statically derivable upper bound $b_{\varphi,x,s}$ on its total $L_1$-sensitivity given an $s$-bounded change to the inputs.
This upper bound can be computed recursively along the dependency graph.
For example, the bound for an input stream is given by $s$, the bound for an addition is the sum of the bounds of its operands, and the bound for a synchronous access coincides with the bound of the accessed stream.
Precise bounds for each operator and the full proof can be found 
\ifthenelse{\boolean{fullversion}}{%
in~\Cref{app:proof_bounds}.
}{%
in the full version of this paper~\cite{fullversion}.
}

With this static bound available, differential privacy can be achieved by injecting appropriately scaled noise into each value of an output stream.
\begin{theorem}
    Given a well-typed specification $\varphi$ with an acyclic dependency graph without unbounded hold operators and a distance value $s$, and $\world \in \World$ a valid evaluation model for $\varphi$.
    Let $\randFunc$ be a randomized mechanism $\randFunc_x : \World \rightarrow \Stream$
	which produces a noisy version of the output stream $x$ by adding independent Laplace noise to each value:
	\[
		\randFunc_x(\world)	:= \lambda t. \world(x)(t) + \eta_{t},
	\]
	where each term $\eta_{t}$ is drawn independently from the Laplace distribution with
	\[
		\eta_t \sim Lap\left(\frac{b_{\varphi,x,s}}{\varepsilon}\right).
	\]
	Then, $\randFunc_s$ satisfies $\varepsilon$-private $s$-distant event-level differential privacy.
\end{theorem}
The result shows that by calibrating Laplace noise to the global sensitivity bound $b_{\varphi,x,s}$, we can guarantee $\varepsilon$-differential privacy for all $s$-distant, event-level adjacent executions.
The proof follows immediately from the combination of \Cref{thm:delta} and \Cref{thm:bound} with \Cref{thm:laplace} and can be found
\ifthenelse{\boolean{fullversion}}{%
in~\Cref{app:proof_dp}.
}{%
in the full version of this paper~\cite{fullversion}.
}
While the randomized mechanism described above releases only a single output stream, standard composition theorems of differential privacy ensure that multiple output streams can also be released with appropriately adjusted privacy parameters.
\section{Accurate Monitoring under Differential Privacy}\label{sec:accuracy}

Differential privacy of individual streams can be ensured by appropriately adding noise to their values.
Simply adding noise to all inputs would guarantee privacy, but it can lead to poor utility of the results.
In this section, we present strategies to maintain accurate monitoring while preserving privacy. 

\subsection{Ensuring Privacy for Unbounded Specifications}\label{sec:segments}

\begin{figure}[t]
\begin{subfigure}[b]{0.4\linewidth}
	\centering
	\begin{tikzpicture}[
		font=\footnotesize,
		n/.style={minimum width=2.5mm,inner sep=0.5mm,font=\scriptsize,draw,circle},
		elabel/.style={font=\scriptsize}
	]
		\node[n] (a) {};
		\node[n,right=5mm of a] (b) {};
		\node[n] (c) at ($(a)!0.5!(b) + (0,-6.5mm)$) {};
		\draw[->] (c) -- (a);
		\draw[->] (c) -- (b);
		\node[n,below=3mm of c] (d) {};
		\draw[->] (d) -- (c);
		\node[n,below=4mm of d,xshift=-2.5mm] (e) {};
		\draw[->] (e) -- (d);
		\node[n,below=4mm of d,xshift=2.5mm] (f) {};
		\draw[->] (f) -- (d);
		\draw[->] (f) to[loop right] (f);
		\draw[->] (e) to[loop left] (e);
		\node[n] (g) at ($(f)!0.5!(e) - (0,6.5mm)$) {};
		\draw[->] (g) -- (e);
		\draw[->] (g) -- (f);
		\begin{scope}[on background layer]
			\node[fit=(a) (b),fill=lightgray!50!white,rounded corners,minimum width=2cm] (inputs) {};
			\node[fit=(c) (d),fill=green!20!white,rounded corners,minimum width=2cm] (loop-free) {};
			\node[fit=(e) (f),fill=lightgray!50!white,rounded corners,minimum width=2cm] (loop) {};
			\node[fit=(g),fill=lightgray!50!white,rounded corners,minimum width=2cm] (public) {};
		\end{scope}
		\node[left=1mm of inputs] {Inputs};
		\node[left=1mm of loop-free,align=right] {Private\\Segment};
		\node[left=1mm of loop,align=right] {Post-Processed\\Segment};
		\node[left=1mm of public,align=right] {Public Outputs};
	\end{tikzpicture}
	\caption{Different segments of the example specification.}
	\label{fig:graph_components}
\end{subfigure}
\hfill
\begin{subfigure}[b]{0.45\linewidth}
	\centering
	\begin{tikzpicture}[
		font=\footnotesize,
		n/.style={minimum width=3mm,inner sep=0.5mm,draw,circle},
		elabel/.style={font=\footnotesize},
		sa/.style={shorten >=1.5mm}
	]
	\foreach \i in {0,1,2} {
		\begin{scope}[xshift=\i*1.6cm]
			\node[n,dotted] (a) {};
			\node[n,dotted,right=4mm of a] (b) {};
			\node[n] at ($(a)!0.5!(b) + (0,-9mm)$) (c) {};
			\node[n,below=6mm of c] (d) {};
			\node[n,dotted,below=6mm of d,xshift=-3mm] (e) {};
			\node[n,dotted,below=6mm of d,xshift=3mm] (f) {};
			\ifthenelse{\i=0}{
				\draw[line width=0.8mm,red] ($(a) + (-1mm,-2.2mm)$) -- ++(3mm,0);
				\draw[line width=0.8mm,red] ($(b) + (1mm,-2.2mm)$) -- ++(-3mm,0);
				\draw[->,sa,shorten >=1.5mm] (c) -- (a);
				\draw[->,sa,shorten >=1.5mm] (c) -- (b);
			}{
				\draw[->] (c) -- (a);
				\draw[->] (c) -- (b);
			}
			\ifthenelse{\i=1}{
				\draw[line width=0.8mm,red] ($(c) + (-2mm,-2.2mm)$) -- ++(4mm,0);
				\draw[->,sa] (d) -- (c);
			}{
				\draw[->] (d) -- (c);
			}
			\ifthenelse{\i=2}{
				\draw[line width=0.8mm,red] ($(d) + (-2mm,-2.2mm)$) -- ++(4mm,0);
				\draw[->,sa] (e) -- (d);
				\draw[->,sa] (f) -- (d);
			}{
				\draw[->,shorten >=0.4mm] (e) -- (d);
				\draw[->,shorten >=0.4mm] (f) -- (d);
			}
		\end{scope}
	}
	\end{tikzpicture}
	\caption{Possible placement options for privacy barriers in the private segment.}
	\label{fig:cut_points}
\end{subfigure}
\caption{Dependency Graph examples for the RTLola specification in~\Cref{fig:example_rtlola}.}
\end{figure}
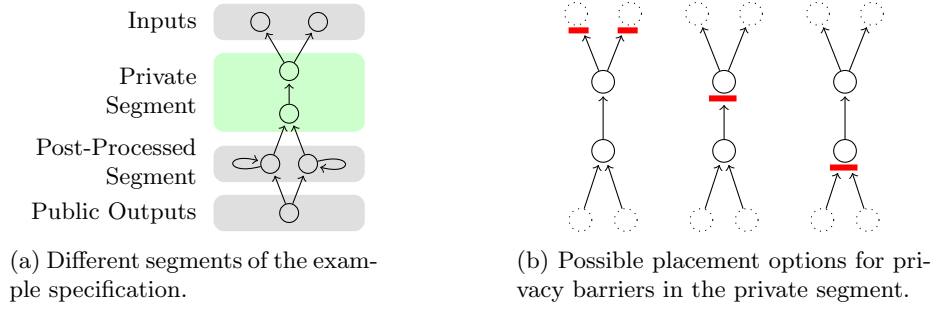

Certain language features, such as unbounded hold operators or circular definitions, can cause a single input event to influence an unbounded number of outputs. 
This makes it impossible to compute a finite global sensitivity for the entire specification directly.
However, every specification can still be made differentially private by controlling where in the specification noise is injected. 
The key idea is to partition the dependency graph of the specification into two segments:

\begin{itemize}
    \item An upper \emph{private segment}, which is acyclic and contains only computations for which the sensitivity can be statically bounded, and
    \item A lower \emph{post-processed segment}, which may include features that induce unbounded sensitivity.
\end{itemize}

The segments are illustrated in \Cref{fig:graph_components} for the specification from \Cref{fig:example_rtlola}.
Noise can be safely injected at the boundaries between the segments, ensuring that all values entering the post-processed segment already satisfy differential privacy. 
Downstream computations, no matter how complex or unbounded, can then be treated as post-processing, which cannot weaken the privacy guarantees.

\subsection{Privacy Barriers}\label{sec:priv_barriers}

Adding noise at the transition between private and post-processed segments ensures differential privacy for any specification.
However, it is not always optimal: injecting noise earlier in the specification can significantly improve the utility of the outputs. 
To formalize where noise can be injected into the specification, we introduce the notion of \emph{privacy barriers}, which separate private inputs from the public outputs in the dependency graph:
\begin{definition}[Privacy Barriers]
Let $\Iref$ be the set of input streams and $\POref \subseteq \Oref$ the set of public output streams. 
A set $\mathcal{C} \subseteq \Iref \uplus \Oref$ is a valid set of \emph{privacy barriers} if every path in the dependency graph from any $i \in \Iref$ to any $o \in \POref$ passes through exactly one $c \in \mathcal{C}$.
\end{definition}
Any valid set of privacy barriers is a valid option for adding noise.
At each barrier, noise is injected, and all downstream computations are treated as post-processing. 
\Cref{fig:cut_points} shows all valid options for the example specification.

A natural goal is to choose a barrier placement that maximizes the utility of the monitored outputs.
However, such an optimal placement does not exist: different choices can favor short-term vs. long-term accuracy, or prioritize the accuracy of different output streams against others.
A placement that is locally optimal for one output stream may be suboptimal for another, leading to inherent trade-offs between streams.
We therefore do not prescribe a single optimal strategy, but instead suggest several heuristics for selecting privacy barriers whose trade-offs can be predicted and reasoned about.

The \emph{input-only} heuristic injects noise at input streams, ensuring that all computations operate on perturbed data.
The \emph{deep} heuristic targets the transition from private to post-processed segments.
This preserves the accuracy of intermediate computations as much as possible, but may require a lot of barriers and, therefore, more noise due to the composition theorem.
The \emph{post-aggregation} heuristic places barriers immediately after the first aggregation.
Compared to adding noise before aggregations, this often reduces the overall noise required since the private input is combined with multiple others before being perturbed.
The \emph{minimal barriers} heuristic selects the smallest set of valid barriers to reduce the number of noise injections. 
Fewer injection points reduce runtime cost but may lead to a higher variance compared to more fine-grained barrier placements.
Importantly, all placement strategies preserve the same differential privacy guarantees.

Beyond these, other heuristics are conceivable, for example, tracking the propagated error symbolically through the specification, similar to approaches in monitoring under uncertainty~\cite{DBLP:conf/rv/FinkbeinerFKK24}.

\subsection{Tree-Based Aggregations}\label{sec:tree_aggregations}

\begin{figure}[t]
	\begin{subfigure}[b]{0.43\linewidth}
	\begin{tikzpicture}[
			xscale=0.9,
			shorten >=0,
			n/.style={draw,circle,inner sep=0.5mm,fill=white},
			m/.style={very thick,n,inner sep=1.1mm,fill=blue!50!cyan,draw=none}
		]
		\foreach \x/\i in {1/0.6,2/2.4,3/3.3,4/5.5,5/6.8} {
			\node[fill=black,inner sep=0.5mm,circle] (ltip\x) at (\i*0.6,0) {};
			\node[n] (tip\x) at ([yshift=3mm]ltip\x) {};
		}

		\draw[->,semithick] (-0.1,0) -- (5.2,0);

		\node[n] (t1) at ($(tip1)!0.5!(tip2) + (0,5mm)$) {};
		\draw (tip1) -- (t1);
		\draw (tip2) -- (t1);
		\node[n] (t2) at ($(tip3)!0.5!(tip4) + (0,5mm)$) {};
		\draw (tip3) -- (t2);
		\draw (tip4) -- (t2);
		\node[n] (t11) at ($(t1)!0.5!(t2) + (0,5mm)$) {};
		\draw (t1) -- (t11);
		\draw (t2) -- (t11);
		\draw[dotted] (tip5) -- ++(5mm,4.8mm) node[n,draw=gray,solid] (g1) {};
		\draw[dotted] (g1) -- ++(3mm,-2mm);

		\begin{scope}[on background layer]
		\node[fit=(ltip1) (tip1),fill=cyan!40!white,rounded corners] {};
		\node[fit=(ltip2) (tip2),fill=cyan!40!white,rounded corners] {};
		\node[fit=(ltip3) (tip3),fill=cyan!40!white,rounded corners] {};
		\node[fit=(ltip4) (tip4),fill=cyan!40!white,rounded corners] {};
		\node[fit=(ltip5) (tip5),fill=cyan!40!white,rounded corners] {};
		\node[m] at (tip5) {};
		\node[m] at (t11) {};
		\end{scope}

		\fill[fill=black] ($(tip5) + (0,-5.5mm)$) -- ++(1mm,-1.5mm) -- ++(-2mm,0) -- cycle;

	\end{tikzpicture}	
	\vspace{2.4mm}
	\caption{Discrete aggregation over all values.}
	\label{fig:tree_all}
	\end{subfigure}
	\hfill
	\begin{subfigure}[b]{0.5\linewidth}
	\begin{tikzpicture}[
			xscale=0.9,
			shorten >=0,
			n/.style={draw,circle,inner sep=0.5mm,fill=white},
			m/.style={very thick,n,inner sep=1.1mm,fill=blue!50!cyan,draw=none}
		]
		\foreach \i in {0.6,2.5,3.2,6.4,6.8,7.2,9.3} {
			\node[fill=black,inner sep=0.5mm,circle] at (\i*0.5,0) {};
		}

		\draw[->,semithick] (-0.1,0) -- (6.25,0);
		\foreach \i in {2,3,4,5} {
			\draw[very thick,shorten >=0,xshift=-1cm,yshift=1.2mm,opacity=0] ($({\i+0.05},0)$) coordinate(x1\i) |- ($({\i+1-0.05},-0.3)$) coordinate(x2\i) -- ($({\i+1-0.05},0)$) coordinate(x3\i);
			\begin{scope}[on background layer]
			\fill[fill=cyan!40!white] (x1\i) |- (x2\i) -- (x3\i) -- ($(x1\i)!0.5!(x3\i) + (0,3mm)$) coordinate(tip\i);
			\end{scope}
			\node[n] (tip\i) at (tip\i) {};
		}

		\def\i{1}
		\draw[very thick,shorten >=0,xshift=-1cm,yshift=1.2mm,opacity=0] ($({\i+0.05},0)$) coordinate(x1\i) |- ($({\i+1-0.05},-0.3)$) coordinate(x2\i) -- ($({\i+1-0.05},0)$) coordinate(x3\i);
		\begin{scope}[on background layer]
		\fill[lightgray] (x1\i) |- (x2\i) -- (x3\i) -- ($(x1\i)!0.5!(x3\i) + (0,3mm)$) coordinate(tip\i);
		\end{scope}
		\node[n] (tip\i) at (tip\i) {};

		\def\i{6}
		\draw[very thick,shorten >=0,xshift=-1cm,yshift=1.2mm,draw=gray,opacity=0] ($({\i+0.05},0)$) coordinate(x1\i) |- ($({\i+1-0.05},-0.3)$) coordinate(x2\i) -- ($({\i+1-0.05},0)$) coordinate(x3\i);
		\begin{scope}[on background layer]
		\fill[fill=lightgray!50!white] (x1\i) |- (x2\i) -- (x3\i) -- ($(x1\i)!0.5!(x3\i) + (0,3mm)$) coordinate(tip\i);
		\end{scope}
		\node[n,draw=gray] (tip\i) at (tip\i) {};

		\draw[decorate,decoration={brace,mirror,amplitude=1.5mm},transform canvas={yshift=-4mm}] (x11) -- node[below=1.5mm] {\footnotesize 1 s} ([xshift=-1mm]x11 -| x12);
		\draw[decorate,decoration={brace,mirror,amplitude=1.5mm},transform canvas={yshift=-4mm}] (x12) -- node[below=1.5mm] {\footnotesize 4 s} ([xshift=-1mm]x11 -| x16);

		\node[n,yshift=5mm] (t1) at ($(tip1)!0.5!(tip2)$) {};
		\draw (tip1) -- (t1);
		\draw (tip2) -- (t1);
		\node[n,yshift=5mm] (t2) at ($(tip3)!0.5!(tip4)$) {};
		\draw (tip3) -- (t2);
		\draw (tip4) -- (t2);
		\node[n,yshift=5mm,draw=gray] (t3) at ($(tip5)!0.5!(tip6)$) {};
		\draw[dotted] (tip5) -- (t3);
		\draw[dotted] (tip6) -- (t3);
		\node[n,yshift=5mm] (t11) at ($(t1)!0.5!(t2)$) {};
		\draw (t1) -- (t11);
		\draw (t2) -- (t11);

		\begin{scope}[on background layer]
		\node[m] at (tip2) {};
		\node[m] at (t2) {};
		\node[m] at (tip5) {};
		\end{scope}

		\fill[fill=black] ([yshift=-5mm,xshift=0.4mm]x35) -- ++(1mm,-1.5mm) -- ++(-2mm,0) -- cycle;
	\end{tikzpicture}
	\vspace{1mm}
	\caption{
		Sliding window aggregation over 4 seconds.
		A new value is produced every second.
	}
	\label{fig:tree_sliding}
	\end{subfigure}
	\caption{
		Tree-Based Aggregation Examples.
	}
\end{figure}
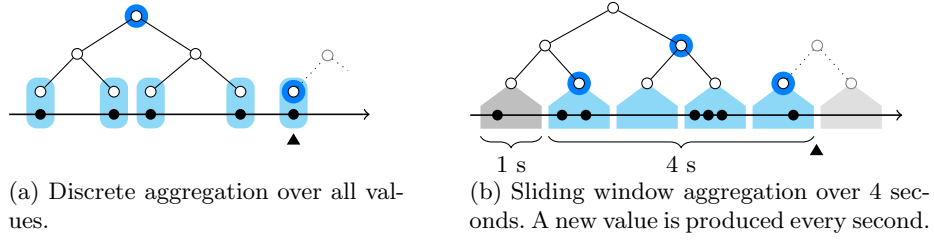

In \Cref{sec:segments}, we showed that we exclude cycles in the dependency graph from the privacy segment as they introduce unbounded sensitivity.
In practice, however, circular definitions are most commonly used to aggregate all values of a stream over time.
Beyond sliding-window aggregations, many stream-based languages support discrete all-aggregations, which aggregate all stream values since the start of the monitor.
The most common loops in specifications can therefore automatically be rewritten as all-aggregations.


To introduce a barrier at an all-aggregation, noise must be added to each released aggregate in order to preserve differential privacy.
However, independently perturbing each aggregated result leads to unbounded privacy loss, since each private value contributes to all future aggregations.
A standard solution to this problem is provided by tree-based mechanisms for differential privacy under continual observation~\cite{DBLP:conf/stoc/DworkNPR10,DBLP:journals/tissec/ChanSS11,DBLP:conf/ndss/PerrierAK19,DBLP:conf/ccs/0001C0SC0LJ21,DBLP:journals/pvldb/QardajiYL13,DBLP:conf/aistats/Cardoso022}.
In these approaches, noisy intermediate aggregates are precomputed and reused across multiple releases.
Rather than repeatedly accessing the private data, aggregates are constructed from these perturbed subaggregates.
Concretely, the stream is partitioned into intervals organized in a binary tree, where each node represents an aggregate over a fixed interval and is perturbed only once.
Aggregates at any time step are obtained by combining the noisy values of a small number of nodes in the tree whose intervals form a partition of the required range.
Several approaches have been proposed to handle unbounded time with these tree-based mechanisms~\cite{DBLP:journals/tissec/ChanSS11,DBLP:conf/ccs/0001C0SC0LJ21}.
In this work, we adopt a construction based on a geometric series that distributes the privacy budget across the infinite height of the tree.

Particular to our setting is that aggregations are not restricted to input streams, where only a single event must be hidden.
Aggregations may also be applied to output streams, where temporal operators have already propagated a single change to multiple stream values.
Our tree-based construction accommodates this, as it only reasons about the total $L_1$-sensitivity of the aggregated stream.
Formal details and proofs are provided
\ifthenelse{\boolean{fullversion}}{%
in \Cref{app:trees}.
}{%
in the full version of this paper~\cite{fullversion}.
}

\Cref{fig:tree_all} illustrates this construction for a discrete all-aggregation.
Each leaf of the tree corresponds to a value of the aggregated stream.
At the position marked with the black triangle, the aggregation over all values is computed by combining the precomputed noisy aggregate over the first four values (the blue node at the top of the tree) with the noisy aggregate of the fifth value at the current position.

In contrast to all-aggregations, whose aggregation is done over a growing prefix of the stream, sliding-window aggregations operate over a time window that shifts with the current evaluation point.
The pacing type requirement in \Cref{def:valid_pacing} ensures that aggregation results are released at a fixed frequency.
This restriction was originally introduced to guarantee that monitoring can be performed using finite memory.
To this end, values are first aggregated into buckets, where the bucket size is determined by the aggregating stream's frequency and the length of the window~\cite{Schwenger_2022}.
We can improve the utility of sliding window aggregations by adopting the tree-based approach, where the buckets form the leaves of the tree, as illustrated in \Cref{fig:tree_sliding}.
To compute, for example, the aggregation over the last four seconds, one sums the noisy values in the blue-marked nodes of the tree.
Using this tree-based construction, in the example, each private value is only released three times, compared to four times without the tree-based approach when each aggregated result is perturbed independently.

\section{Stricter Privacy Guarantees}\label{sec:stricter}

Our approach assumes that the values at the individual input stream positions are independent,  but this assumption does not hold in all applications.
For example, consecutive sensor readings, such as GPS coordinates, are often correlated.
In such cases, standard event-level differential privacy underestimates the required amount of noise.
One way to address this is to adopt $w$-consecutive differential privacy ($w$-DP)~\cite{DBLP:journals/pvldb/KellarisPXP14}, which hides the effect of up to $w$ consecutive input values.
If correlations are limited to $w$ consecutive values, w-DP provides meaningful privacy guarantees.
Extending our approach to $w$-DP is straightforward, since our mechanisms already handle correlated outputs: we can apply the same techniques to $w$-sized groups of input values.

Another complementary strategy is to use a preprocessing monitor.
In this approach, a first monitor transforms correlated inputs into independent values or aggregated summaries. These independent streams are then fed into a second, privacy-preserving monitor that ensures differential privacy.
For example, the first monitor could extract individual trips from raw GPS data.
The second monitor then computes the private outputs over these independent summaries.
This strategy of using a preprocessing monitor can even ensure user-level privacy~\cite{DBLP:conf/stoc/DworkNPR10} by ensuring the second privacy-preserving monitor receives a single event per user.
\section{Implementation and Evaluation}\label{sec:evaluation}

This section evaluates the practical applicability of our approach.
We first describe the implementation of the privacy analysis within the RTLola framework and then assess the impact of different design choices using synthetic examples and a case study on monitoring public transportation data.

\subsection{Implementation}

\begin{figure}[t]
	\begin{minipage}[b]{0.45\linewidth}
		\begin{tabular}{lrr}
\toprule
Specification & no-priv & priv \\
\midrule
Fraud Detection \ifthenelse{\boolean{fullversion}}{(\ref{app:fraud_detection})}{} & 3\,ms & 7\,ms \\
Geofence Small \cite{DBLP:conf/cav/BaumeisterFKLMST24} & 118\,ms & 278\,ms \\
Geofence Large \cite{DBLP:conf/cav/BaumeisterFKLMST24} & 594\,ms & 1368\,ms \\
Peak Detection \cite{DBLP:conf/cav/BaumeisterFSST20} & 3\,ms & 8\,ms \\
Public Transport \ifthenelse{\boolean{fullversion}}{(\ref{app:public_transportation})}{} & 5\,ms & 12\,ms \\
Driving Emissions \cite{DBLP:conf/tacas/BiewerFHKSS21} & 123\,ms & 247\,ms \\
Sensor Validation \cite{DBLP:conf/cav/BaumeisterFSST20} & 2\,ms & 7\,ms \\
\bottomrule
\end{tabular}

		\vspace{4mm}
		\captionof{table}{Analysis runtime comparison for different specifications.}
		\label{tab:analysis_runtime}
	\end{minipage}
	\hfill
	\begin{minipage}[b]{0.495\linewidth}
		\includegraphics[width=\linewidth,alt={
			A line plot showing the window size on the x axis (1 to 15) and the variance of the davg stream on the y axis (log scale).
			It compares the only-inputs, regular, and tree-based methods for different numbers of values per bucket (1, 10, 100), with the lines spreading apart as the window size increases.
		}]{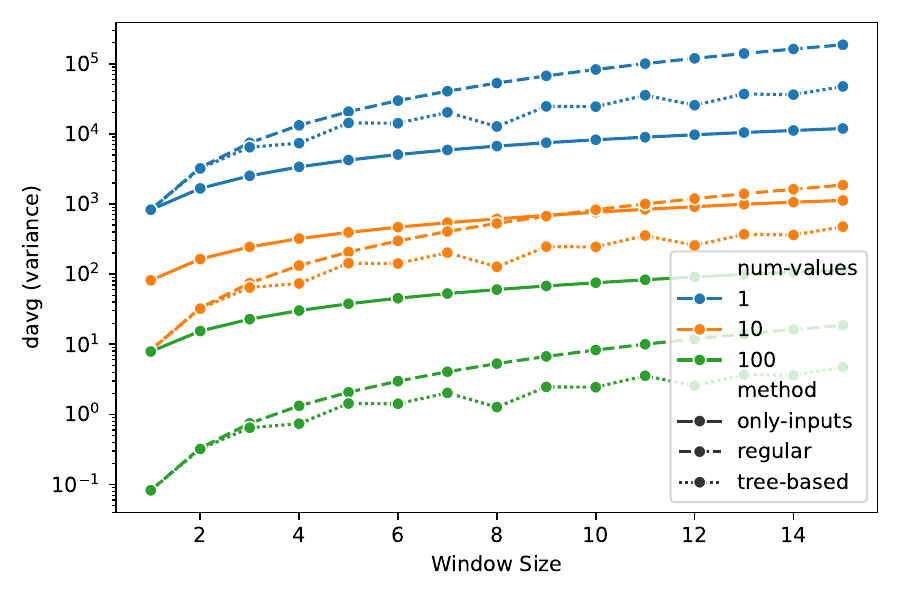}	
		\captionof{figure}{Output variances for different sliding window aggregations.}
		\label{fig:aggregation_utility}
	\end{minipage}
\end{figure}

Our implementation extends the RTLola frontend~\cite{DBLP:conf/fm/BaumeisterFKS24} with a privacy analysis that is executed after the dependency analysis.
This analysis identifies all valid privacy barriers in the dependency graph and applies the heuristics described in \Cref{sec:priv_barriers} to select the barriers at which noise is injected.
Using propagated sensitivity and range information, noise expressions are inserted into the intermediate representation of the specification, and tree-based aggregation is translated into corresponding RTLola constructs.
This way, no modifications to the existing backend implementations~\cite{DBLP:conf/cav/BaumeisterCFS25,DBLP:conf/cav/FaymonvilleFSSS19} are required for private monitoring.

We evaluate the runtime of the privacy analysis.
The results are summarized in \Cref{tab:analysis_runtime}, where the \emph{priv} column reports the maximum runtime across all heuristics.
The measurements were obtained on a system with a 13th Gen Intel Core i7-1355U processor over 50 runs.
We observe that the runtime more than doubles for all specifications.
This increase is primarily due to the fact that, following the privacy analysis, which introduces new expressions and streams in the specification, we rerun earlier stages such as dependency analysis and type checking.
Despite this overhead, all runtimes remain well below a few seconds, rendering the additional cost negligible for an analysis performed once prior to monitor deployment.

\subsection{Impact of Tree-Based Aggregations}

Next, we assess the impact of the tree-based aggregation on the accuracy of the output using the example shown in \Cref{fig:example_rtlola}.
Accuracy is quantified by computing, for each timestamp, the variance across 100 runs of the private monitor ($\varepsilon = 1$) on the same input trace and then averaging this variance across all timestamps.
Lower variance indicates higher accuracy of the output.

We compare how accuracy is influenced by the structure of the aggregation window.
First, we vary the number of buckets in the window.
Increasing the number of buckets causes individual values to contribute to multiple aggregates, which in turn requires more perturbation to ensure privacy.
Second, we vary the number of values per bucket.
More values per bucket allow individual values to be more effectively obscured by others, thereby reducing the variance.

We compare three perturbation strategies.
As a baseline, we add noise directly to each input before aggregation.
We then consider the basic sensitivity-based perturbation described in \Cref{sec:dp_lola} and the tree-based aggregation approach from \Cref{sec:tree_aggregations}.
\Cref{fig:aggregation_utility} shows that for small windows, all three methods achieve similar utility.
As the window size grows, the differences become more pronounced, with the tree-based approach yielding significantly lower variance than the sensitivity-based method.
When each bucket contains only a single value, directly perturbing the inputs yields the highest accuracy.
However, this approach rapidly degrades in utility as the number of values per bucket increases.

\subsection{Case Study}

\begin{figure}[t]
	\begin{subfigure}{0.495\linewidth}
	\includegraphics[width=\linewidth,alt={
		Three histograms stacked vertically for each bus line: night bus, university bus and city bus.
		All show the time of day at the x-axis and the number of events on the y-axis.
		The night bus has fewer number of events (< 16 per bucket) , and only at nighttime hours.
		The university bus has a lot of events (> 60 per bucket) between 7:00 and 18:00.
		The city bus also has a lot of events (> 50 per bucket) during daytime hours, with a peek at noon.
	}]{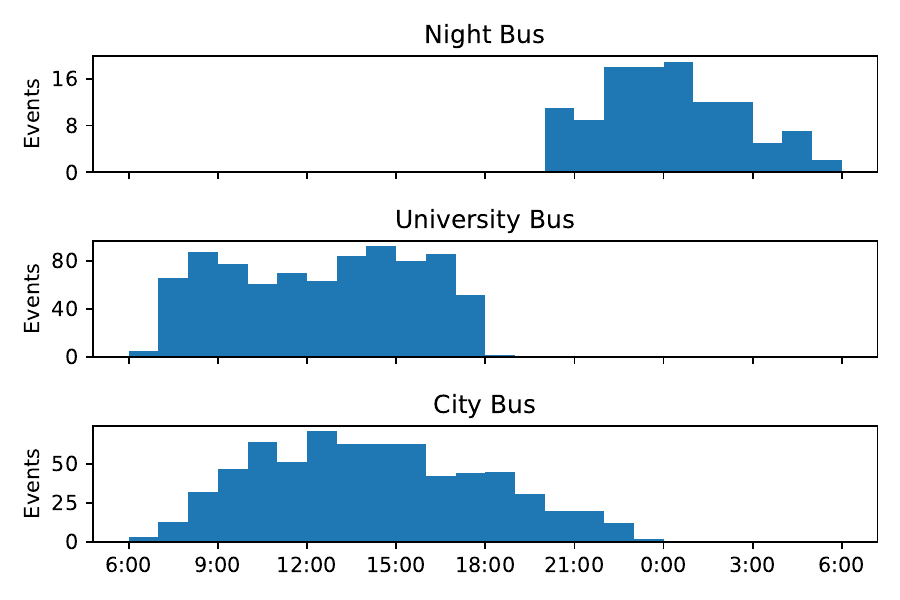}	
	\caption{Histogram of the number of events per hour per bus line.}
	\label{fig:case_study_histogram}
	\end{subfigure}	
	\hfill
	\begin{subfigure}{0.495\linewidth}
	\includegraphics[width=\linewidth,alt={
		A line plot showing the average crowdedness of three bus lines (night bus, university bus, and city bus) over the hour of the day.
		For university and city bus the peak is at daytime hours, and for the party bus the peak is at midnight.
		The plot has an error bar which is small for the university bus, slightly larger for the city bus, and large for the night bus.
	}]{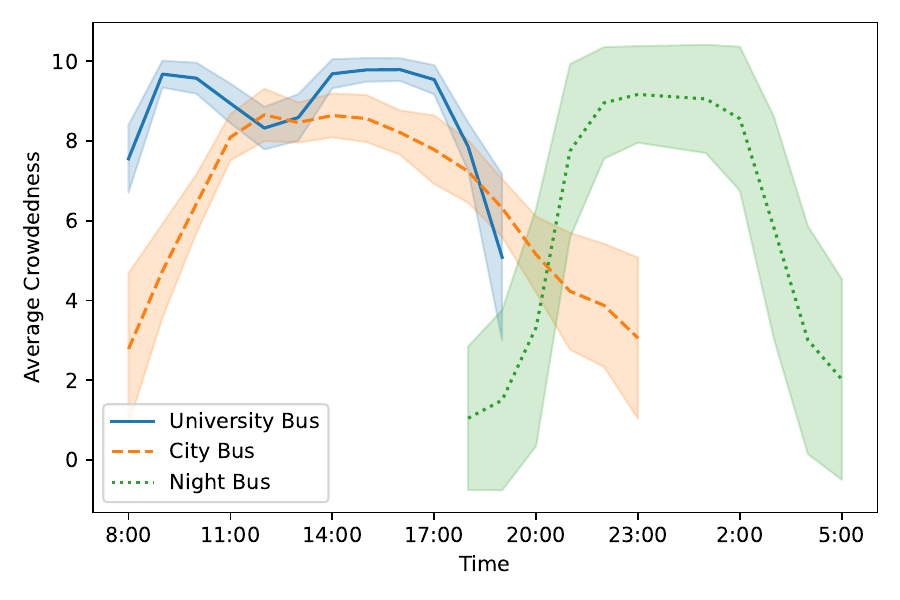}	
	\caption{The average crowdedness per hour output by the differentially private monitor.}
	\label{fig:average_crowdedness}
	\end{subfigure}	
	\caption{
		Results from monitoring synthetic public transportation data.
	}
\end{figure}

To demonstrate the applicability of our approach in a realistic setting, we present a case study from the public transportation domain.
We implemented a system that crowdsources information about public transportation usage by allowing users to report trips on specific bus lines.
The reported events are processed by an RTLola monitor, which aggregates the data and displays anonymized usage statistics to users.
The aggregated results provide insights into temporal usage patterns, such as the expected level of crowding on a bus line at different times of day.
Since the collected data reflects individual movement behavior, it is inherently privacy-sensitive.
By relying on our privacy-preserving framework, all privacy guarantees are enforced automatically by the frontend analysis, allowing the application to be implemented without manual reasoning about privacy concerns.
\ifthenelse{\boolean{fullversion}}{
    The corresponding RTLola specification is provided in \Cref{app:public_transportation}.
}{
    The corresponding RTLola specification is provided in the full version of this paper~\cite{fullversion}.
}

We evaluate our case study using synthetic data.
Specifically, we consider three bus lines and generate a varying number of users depending on the time of day, as shown in \Cref{fig:case_study_histogram}.
The first is a university bus, highly frequented but only during a shorter period of the day.
The second is a city bus, which is frequently used throughout daytime hours.
The third is a night bus, which also exhibits high usage per bus, but with significantly fewer buses and users overall.

In \Cref{fig:average_crowdedness}, we show the average crowdedness per time of day as reported by the private monitor ($\varepsilon = 1$, deep heuristic), with the standard deviation (SD) shown as error bars.
The SD indicates that the confidence of the monitor output depends on the number of users.
To estimate the SD introduced by differential privacy, we execute the monitor 1500 times and calculate statistics over these runs.
During high-traffic daytime hours, the differentially private estimates closely match the true crowdedness.
At night, the smaller number of participants leads to a higher SD.
But even during nighttime hours with fewer participants, the monitor output closely approximates the actual crowdedness.
\section{Conclusion}

In this paper, we have presented a method for integrating differential privacy into runtime monitoring to ensure system correctness while preserving the privacy of sensitive system data.
Our approach analyzes the specifications and automatically enforces privacy guarantees without requiring any user interaction.
As a result, users without privacy expertise can benefit from differential privacy without having to reason about sensitivities, noise calibration, or privacy budgets, which is particularly challenging given temporal dependencies in stream-based specifications.
Finally, privacy is enforced in a utility-aware manner by injecting noise at carefully selected points in the specification, thereby preserving the accuracy of monitoring outcomes while still providing rigorous privacy guarantees.
We have implemented our approach in the stream-based monitoring framework RTLola, demonstrating that privacy-preserving runtime monitoring can be achieved automatically and with minimal loss of accuracy.

Our work lays the foundation for deploying runtime monitoring in privacy-critical domains where such techniques have traditionally been difficult to apply.
This includes areas such as medical or financial systems, where continuous monitoring is essential, but the underlying data is highly privacy-critical.
We believe that making privacy guarantees an integral part of runtime monitoring is a step toward making runtime verification applicable in a broader range of real-world, safety-critical and privacy-sensitive systems.

\subsubsection{Acknowledgments.} This work was partially supported by the German Research Foundation (DFG) as part of TRR 248 (No.~389792660) and PreCePT (No.~521273327), and by the European Research Council (ERC) Grant HYPER (No.~101055412).

\subsubsection{Disclosure of Interests.}
The authors have no competing interests to declare that are relevant to the content of this article.

\subsubsection{Data-Availability Statement.}
The artifacts and resources associated with this article are accessible via
\url{https://doi.org/10.5281/zenodo.19694096}.

\bibliographystyle{splncs04}
\bibliography{bibliography.bib}

\ifthenelse{\boolean{fullversion}}{
    \clearpage
    \appendix
    \section{Operator Semantics}\label{app:semantics}

Given an evaluation model $\world$ and a corresponding pacing model $p$, we define the semantics of the operators as follows.
Note that for the semantics and the following proofs, we assume there is a separate stream definition for each subexpression.
\begin{align*}
	\world(e_1 + e_2)(t) &= \world(e_1)(t) + \world(e_2)(t)\\
	\world(e_1 \cdot e_2)(t) &= \world(e_1)(t) \cdot \world(e_2)(t)\\
	\world(\sync(y))(t) &= \world(y)(t)\\
	\world(\offset(y, o))(t) &= \begin{cases}
		\world(y)(\olast(y, t, o, p)) & \text{if $\olast(y, t, o, p)$ exists}\\
		\bot & \text{otherwise}
	\end{cases}\\
	\world(\mathsf{hold}(y))(t) &= \begin{cases}
		\world(y)(\olast(y, t, 0, p)) & \text{ if $\olast(y, t, 0, p)$ exists}\\
		\bot & \text{otherwise}
	\end{cases}\\
	\world(\mathsf{defaults}(e, dft))(t) &= \begin{cases}
		\world(e)(t) & \text{if $\world(e)(t) \ne \bot$}\\
		\world(dft)(t) & \text{otherwise}
	\end{cases}\\
	\world(\mathsf{ite}(c, e_1, e_2))(t) &= \begin{cases}
		\world(e_1)(t) & \text{if $\world(c)(t) = \top$}\\
		\world(e_2)(t) & \text{otherwise}
	\end{cases}\\
	\world(\mathsf{aggr}(y, W, f))(t) &= f(\setCond{\world(y)(t')}{t' \in \windowtimes(y, t, p, W)})\\
	\world(\mathsf{hold}_{bound}(y))(t) &= \begin{cases}
		\world(y)(\olast(y, t, 0, p)) & \substack{\text{if $\olast(y, t, 0, p)$ exists}\\\land holdn(x, y, t, p) \le bound}\\
		\bot & \text{otherwise}
	\end{cases}
\end{align*}
We use the following predicate to determine how often a specific value of $y$ is accessed from $x$:
\begin{align*}
	holdn(x, y, t, p) &= \left|\setCond{t' \in \time}{p(x)(t') \land t' \le t \land \olast(y, t', 0, p) = \olast(y, t, 0, p)}\right|
\end{align*}

\clearpage
\section{Definition $\Delta$}\label{app:delta}

Given stream expressions $e$, $e_1$ and $e_2$, the constant $c$, and the stream reference $y$, we define
    \begin{align*}
    \Delta_{t,p,s}^{t'}(e_1 + e_2) &= \Delta_{t,p,s}^{t'}(e_1) + \Delta_{t,p,s}^{t'}(e_2)\\
    \Delta_{t,p,s}^{t'}(c \cdot e) &= |c| \cdot \Delta_{t,p,s}^{t'}(e)\\
    \Delta_{t,p,s}^{t'}(c) &= 0\\
    \Delta_{t,p,s}^{t'}(\mathsf{sync}(y)) &= \Delta_{t,p,s}^{t'}(y)\\
    \Delta_{t,p,s}^{t'}(\mathsf{offset}(y, o)) &= \Delta_{t,p,s}^{\olast(y, t', p, o)}(y)\\
    \Delta_{t,p,s}^{t'}(\mathsf{aggr}(y, W, \Sigma)) &= \sum_{t'' \in \windowtimes(y, t', p, W)} \Delta_{t,p,s}^{t''}(y)\\
	\Delta_{t,p,s}^{t'}(\mathsf{aggr}(y, W, last)) &= \begin{cases}
		\Delta_{t,p,s}^{\olast(y, t', 0, p)}(y) & \text{if $\olast(y, t', 0, p) \ge t' - W$}\\
		0 & \text{otherwise}
	\end{cases}\\
    \Delta_{t,p,s}^{t'}(\mathsf{aggr}(y, W, \#)) &= 0\\
	\Delta_{t,p,s}^{t'}(\mathsf{hold}(y)) &= \Delta_{t,p,s}^{\olast(y, t', 0, p)}(y)\\
	\Delta_{t,p,s}^{t'}(\mathsf{hold}_{bound}(y)) &= \begin{cases}
		\Delta_{t,p,s}^{\olast(y, t', 0, p)}(y) & \text{if } holdn(x, y, t', p) \le bound\\
		0 & \text{otherwise}
	\end{cases}\\
	\Delta_{t,p,s}^{t'}(\mathsf{default}(e, dft)) &= \max(\Delta_{t,p,s}^{t'}(e), \Delta_{t,p,s}^{t'}(dft))\\
\Delta_{t,p,s}^{t'}(x) &= 
\begin{cases} 
ub(x) - lb(x) & \text{if $t \rightsquigarrow_p^x t'$} \\
0 & \text{otherwise,}
\end{cases}\quad\substack{\text{for value-dependent}\\\text{operators}}
\end{align*}
with
\begin{align*}
	\olast(y, t, p, o) &= \max\nolimits_o \setCond{t' \in \time}{t' < t \land p(y)(t')}\\
	\windowtimes(y, t, p, W) &= \setCond{t' \in \time}{p(t) - W \le p(t') \le p(t) \land p(y)(t')}.
\end{align*}
$\rightsquigarrow$ is defined in \Cref{app:def_influence}.

\section{Bounds of Expressions}\label{app:def_bounds}

The \emph{lower bound} of a stream is defined as
\begin{align*}
	lb : \Iref \uplus \Oref \rightarrow \RR \cup \set{-\infty}\\
	lb(x) = \begin{cases}
		-\infty & \text{if } x \in \Iref\\
		lb(\varphi_x) & \text{if } x \in \Oref,
	\end{cases}
\end{align*}	
and recursively on expressions as
\begin{align*}
	lb(\sync(y)) &= lb(y)\\	
	lb(\offset(y, \_)) &= lb(y)\\	
	lb(f_1 + f_2) &= lb(f_1) + lb(f_2)\\
	lb(f_1 \cdot f_2) &= min(ub(f_1) \cdot ub(f_2), lb(f_1) \cdot ub(f_2),\\
					& \qquad ub(f_1) \cdot lb(f_2), lb(f_1) \cdot ub(f_2))\\
	lb(\mathsf{aggr}(y, W, S, \Sigma)) &= -\infty\\
	lb(\mathsf{aggr}(y, W, S, \#)) &= 0\\
	lb(\mathsf{ite}(c, f_1, f_2)) &= min(lb(f_1), lb(f_2))\\
	lb(\mathsf{clamp}(f, nlb, nub)) &= nlb.
\end{align*}

The \emph{upper bound} of a stream is defined as
\begin{align*}
	ub : \Iref \uplus \Oref \rightarrow \RR \cup \set{\infty}\\
	ub(x) = \begin{cases}
		\infty & \text{if } x \in \Iref\\
		ub(\varphi_x) & \text{if } x \in \Oref,
	\end{cases}
\end{align*}	
and recursively on expressions as
\begin{align*}
	ub(\sync(y)) &= ub(y)\\	
	ub(\offset(y, \_)) &= ub(y)\\	
	ub(f_1 + f_2) &= ub(f_1) + ub(f_2)\\
	ub(f_1 \cdot f_2) &= \max(ub(f_1) \cdot ub(f_2), lb(f_1) \cdot ub(f_2),\\
				& \qquad ub(f_1) \cdot lb(f_2), lb(f_1) \cdot ub(f_2))\\
	ub(\mathsf{aggr}(y, W, S, \Sigma)) &= \infty\\
	ub(\mathsf{aggr}(y, W, S, \#)) &= 0\\
	ub(\mathsf{ite}(c, f_1, f_2)) &= max(ub(f_1), ub(f_2))\\
	ub(\mathsf{clamp}(f, nlb, nub)) &= nub.
\end{align*}

\section{Proof for \Cref{thm:delta}}\label{app:proof_delta}

\begin{proof}
	Proof by structural induction over the defining stream expression of $x$.

	\paragraph{Base Case:} If $x$ is an input stream, then, if $t' = t$, $|\world(x)(t') - \world'(x)(t')| \le s$ by assumption ($s$-distant), and $\Delta_{t,p,s}^{t'}(x) = s$ by definition of $\Delta$.
	If $t' \ne t$, then $|\world(x)(t') - \world'(x)(t')| = 0$ by assumption (event-level adjacent and differ at timestamp $t$), and $\Delta_{t,p,s}^{t'}(x) = 0$ by definition of $\Delta$.

	\paragraph{Induction Hypothesis:} Assume that for all subexpressions $e$ of the stream expression for $x$, for every timestamp $t'$, it holds that
	\[
		|\world(x)(t') - \world'(x)(t')| \le \Delta_{t,p,s}^{t'}(x).
	\]

	\paragraph{Induction Step:} Consider the possible operations in the stream expression:
	\begin{itemize}
    \item \textbf{Addition:} $e = e_1 + e_2$
    \begin{align*}
	|\world(e)(t') - \world'(e)(t')| &= \bigl|(\world(e_1)(t') + \world(e_2)(t'))\\
	&\quad - (\world'(e_1)(t') + \world'(e_2)(t'))\bigr| \\
	&= \bigl|(\world(e_1)(t') - \world'(e_1)(t'))\\
	&\quad+ (\world(e_2)(t') - \world'(e_2)(t'))\bigr|\\
    &\le |\world(e_1)(t') - \world'(e_1)(t')|\\
	& \quad + |\world(e_2)(t') - \world'(e_2)(t')| &&\text{Triangle inequality}\\
	&\le \Delta_{t,p}^{t'}(f_1) + \Delta_{t,p}^{t'}(f_2) && \text{Induction hypothesis}\\
	&= \Delta_{t,p}^{t'}(f) && \text{Def. $\Delta$}
	\end{align*}

    \item \textbf{Scaling by constant:} $e = c \cdot e'$
    \begin{align*}
    |\world(e)(t') - \world'(e)(t')| &= |c \cdot (\world(e')(t') - \world'(e')(t'))| \\
	&= |c| \cdot |\world(e')(t') - \world(e')(t')| && \text{absolute value}\\
	&\le |c| \cdot \Delta_{t,p,s}^{t'}(e') && \text{Induction hypothesis}\\
	&= \Delta_{t,p,s}^{t'}(e) && \text{Def. $\Delta$}
	\end{align*}

    \item \textbf{Offset access:} $e = \mathsf{offset}(y, o)$
    \begin{align*}
    |\world(e)(t') - \world'(e)(t')| &= |\world(y)(\olast(X, t', p, o))\\
	&\quad - \world'(y)(\olast(y, t', p, o))| &&\text{behaves according to $p$}\\
	&\le \Delta_{t,p,s}^{\olast(y, t', o, p)}(X) &&\text{Induction hypothesis}\\
	&= \Delta_{t,p,s}^{t'}(e) && \text{Def. $\Delta$}
	\end{align*}

    \item \textbf{Synchronous access:} $e = \mathsf{sync}(y)$
    \begin{align*}
    |\world(e)(t') - \world(e)(t')| &= |\world(y)(t') - \world'(y)(t')|\\
	&\le \Delta_{t,p,s}^{t'}(y) &&\text{Induction hypothesis}\\
	&= \Delta_{t,p,s}^{t'}(e) &&\text{Def. $\Delta$}
	\end{align*}
    \item \textbf{Window access (sum):} $e = \mathsf{aggr}(y, W, \sigma)$
    \begin{align*}
    &|\world(e)(t') - \world(e)(t')| \\
	&= \Big| \sum_{t'' \in \windowtimes(y, t', p, W)} \world(y)(t'')\\
	&\quad - \sum_{t'' \in \windowtimes(y, t', p, W)} \world'(y)(t'') \Big| && \text{behaves according to $p$}\\
	&= \Big| \sum_{t'' \in \windowtimes(y, t', p, W)} \world(y)(t'') - \world'(y)(t'') \Big| && \text{\shortstack{$\windowtimes$ is finite\\($p$ is well-typed)}}\\
	&\le \sum_{t'' \in \windowtimes(X, t', p, W)} |\world(X)(t'') - \world'(y)(t'') |&&\text{absolute value}\\
	&\le \sum_{t'' \in \windowtimes(y, t', p, W)} \Delta_{t,p,s}^{t'}(y) && \text{Inducation hypothesis}\\
	&= \Delta_{t,p,s}^{t'}(e) && \text{Def. $\Delta$}
	\end{align*}
    \item \textbf{Window access (count):} $e = \mathsf{aggr}(y, W, \#)$
    \begin{align*}
    |\world(e)(t') - \world'(e)(t')| &= \Big| \sum_{t'' \in \windowtimes(y, t', p, W)} 1\\
	&\quad - \sum_{t'' \in \windowtimes(y, t', p, W)} 1\Big| && \text{behaves according to $p$}\\
	&= 0 = \Delta_{t,p,s}^{t'}(e) &&\text{Def. $\Delta$}
	\end{align*}

	\item \textbf{Window access (last):} $e = \mathsf{aggr}(y, W, last)$
	\begin{align*}
	|\world(e)(t') - \world'(e)(t')| &= |\world(y)(\olast(y,t',0,p)) - \world'(y)(\olast(y,t',0,p))\\
	&\le \Delta_{t,p,s}^{\olast(y,t',0,p)}(y) && \hspace{-2.5cm}\text{Induction hypothesis}\\
	&= \Delta_{t,p,s}^{\olast(y,t',0,p)}(e) && \hspace{-2.5cm}\text{Def. $\Delta$}
	\end{align*}

	\item \textbf{Hold:} $e = \mathsf{hold}(y)$. The same as for last aggregation:
	\begin{align*}
	|\world(e)(t') - \world'(e)(t')| &= |\world(y)(\olast(y,t',0,p)) - \world'(y)(\olast(y,t',0,p))\\
	&\le \Delta_{t,p,s}^{\olast(y,t',0,p)}(y) && \hspace{-2.5cm}\text{Induction hypothesis}\\
	&= \Delta_{t,p,s}^{\olast(y,t',0,p)}(e) && \hspace{-2.5cm}\text{Def. $\Delta$}
	\end{align*}
\end{itemize}

\noindent
For the extension of $\Delta$ presented in \Cref{sec:value_dependent}, it holds that
\begin{align*}
	\world(e)(t') \ge lb(e) \qquad \text{and} \qquad \world'(e)(t') \le ub(e).
\end{align*}
From this, it follows that
\begin{align*}
|\world(e)(t') - \world'(e)(t')| \le ub(x) - lb(x).
\end{align*}

\noindent
The recursion ends because the dependency graph is acyclic.
\end{proof}

\section{Definition $\rightsquigarrow$}\label{app:def_influence}

A change in the inputs at time $t$ influences the value of $x$ at time $t'$ if
\begin{align*}
	t \rightsquigarrow^x_p t' &:= \begin{cases}
		t = t' & \text{if $x \in \Iref$}\\
		t \rightsquigarrow^{\varphi_x}_p t' & \text{if $x \in \Oref$}
	\end{cases}\\
	t \rightsquigarrow^{\sync(y)}_p t' &:= t \rightsquigarrow^y_p t'\\
	t \rightsquigarrow^{\offset(y, o)}_p t' &:= t \rightsquigarrow^y_p \olast(y, t', o, p)\\
	t \rightsquigarrow^{\mathsf{hold}(y)}_p t' &:= t \rightsquigarrow^y_p \olast(y, t', 0, p)\\
	t \rightsquigarrow^{e_1 \circ e_2}_p t' &:= t \rightsquigarrow^{e_1}_p t' \lor t \rightsquigarrow^{e_2}_p t' \quad \text{for all binary operators $\circ$}\\
	t \rightsquigarrow^{\mathsf{ite}(c, e_1, e_2)}_p t' &:= t \rightsquigarrow^{c}_p t' \lor t \rightsquigarrow^{e_1}_p t' \lor t \rightsquigarrow^{e_2}_p t'
\end{align*}

\begin{theorem}
	Given an acyclic specification $\varphi$ without unbounded hold operators, then it holds that
	\[
		\exists n_x . \forall p \in \pacingModel_\varphi . \forall t \in \time. |\setCond{t' \in T_e}{t \rightsquigarrow_p^x t'}| \le n_x
	\]
	the number of timepoints of $x$ that are influenced by a change of an input at any timepoint is bounded by a constant $n_x$.
\end{theorem}

\begin{proof}
	Proof by structural induction over stream-expression $x$:

	\paragraph{Base Case:} If $x$ is an input stream, there is exactly one time point $t'$ where $t = t'$.

	\paragraph{Induction Hypothesis:} Assume that for all subexpressions $e$ of the stream expression for $x$, it holds that
	\[
		\exists n_e . \forall p \in \pacingModel_\varphi . \forall t \in T_e. |\setCond{t' \in T_e}{t \rightsquigarrow_p^e t'}| \le n_e.
	\]

	\paragraph{Induction Step:} Consider the possible operations in the stream expression:

	\begin{itemize}
		\item \textbf{Binary operators:} $e = e_1 \circ e_2$. Then we set $n_e = n_{e_1} + n_{e_2}$ and it holds that
		\begin{align*}
			|\setCond{t' \in T_x}{t \rightsquigarrow_p^e t'}| &= |\setCond{t' \in T_x}{t \rightsquigarrow_p^{e_1} t' \lor t \rightsquigarrow_p^{e_2} t'}|\\
			&\le |\setCond{t' \in T_x}{t \rightsquigarrow_p^{e_1} t'}|\\
			&\quad + |\setCond{t' \in T_x}{t \rightsquigarrow_p^{e_2} t'}|\\
			&\le n_{e_1} + n_{e_2} = n_e
		\end{align*}
		\item \textbf{Synchronous access:} $e = \sync(y)$. Then we set $n_e = n_y$ and it holds that
		\begin{align*}
			|\setCond{t' \in T_x}{t \rightsquigarrow_p^e t'}| &= |\setCond{t' \in T_x}{t \rightsquigarrow_p^{y} t'}|\\
			&\le |\setCond{t' \in T_y}{t \rightsquigarrow_p^{y} t'}| && \text{$p$ well-typed ($T_x \subseteq T_y$)}\\
			&\le n_y = n_e
		\end{align*}
		\item \textbf{Offset access:} $e = \offset(y, o)$. Then we set $n_e = n_y$ and it holds that
		\begin{align*}
			|\setCond{t' \in T_x}{t \rightsquigarrow_p^e t'}| &= |\setCond{t' \in T_x}{t \rightsquigarrow_p^{y} \olast(y, t', o, p)}|\\
			&= |\setCond{t' \in T_y}{t \rightsquigarrow_p^{y} t'}|\\
			&= n_y = n_e
		\end{align*}
		Because $p$ is well-typed, for every $t'' \in T_y$, there exists only one $t' \in T_x$ for which $\olast(y, t', o, p) = t''$.
	\end{itemize}
	The proofs for the other operators follow a very similar structure to the proof in \Cref{app:proof_bounds}.
	The recursion ends because the dependency graph is acyclic.
\end{proof}

\section{Proof for \Cref{thm:bound}}\label{app:proof_bounds}

\begin{proof}
	Proof by structural induction over stream-expression $x$:

	\paragraph{Base Case:} If $x$ is an input stream, then we set $b = s$, and it holds that:
	\begin{align*}
		\sum_{t' \in T_x} \Delta_{t,p,s}^{t'} (I) &= \sum_{t' \in T_x} \begin{cases}
			s & \text{if } t' = t\\0 & \text{otherwise}
		\end{cases}\\
		&\le s = b.
	\end{align*}

	\paragraph{Induction Hypothesis:} Assume that for all subexpressions $e$ of the stream expression for $x$, it holds that
	\[
		\exists b_e. \forall t. \left(\sum_{t' \in T_x} \Delta_{t,p,s}^{t'} e\right) \le b_e.
	\]

	\paragraph{Induction Step:} Consider the possible operations in the stream expression:

	\begin{itemize}
		\item \textbf{Addition:} $e = e_1 + e_2$. Then we set $b_e = b_{e_1} + b_{e_2}$ and it holds that
		\begin{align*}
			\sum_{t' \in T_x} \Delta_{t,p,s}^{t'} e &= \sum_{t' \in T_x} \left(\Delta_{t,p,s}^{t'}(e_1) + \Delta_{t,p,s}^{t'}(e_2)\right) &&\text{Def. $\Delta$}\\
			&= \sum_{t' \in T_x} \Delta_{t,p,s}^{t'}(e_1) + \sum_{t' \in T_x} \Delta_{t,p,s}^{t'}(e_2) &&\text{$\Delta$ non-negativ}\\
			&\le b_{e_1} + b_{e_2} &&\text{Induction Hypothesis}\\
			&= b_e
		\end{align*}
		\item \textbf{Scaling by constant:} $e = c \cdot e'$. Then we set $b_e = |c| \cdot b_{e'}$ and it holds that
		\begin{align*}
			\sum_{t' \in T_x} \Delta_{t,p,s}^{t'} e &= \sum_{t' \in T_x} \Delta_{t,p,s}^{t'}(c \cdot e')\\
			&= \sum_{t' \in T_x} |c| \cdot \Delta_{t,p,s}^{t'}(e') && \text{Def. $\Delta$}\\
			&= |c| \cdot \sum_{t' \in T_x} \Delta_{t,p,s}^{t'}(e')\\
			&\le |c| \cdot b_{e'} = b_e &&\text{Induction Hypothesis}
		\end{align*}
	\item \textbf{Synchronous access:} $e = \mathsf{sync}(y)$. Then we set $b_e = b_y$ and it holds that
	\begin{align*}
			\sum_{t' \in T_x} \Delta_{t,p,s}^{t'} e &= \sum_{t' \in T_x} \Delta_{t,p,s}^{t'}(\mathsf{sync}(y))\\
			&= \sum_{t' \in T_x} \Delta_{t,p,s}^{t'}(y) &&\text{Def. $\Delta$}\\
			&= \sum_{t' \in T_y} \Delta_{t,p,s}^{t'}(y) &&\text{$p$ well-typed ($T_x \subseteq T_y$)}\\
			&\le b_y = b_e && \text{Induction Hypothesis}
	\end{align*}
	\item \textbf{Offset access:} $e = \offset(y, o)$. Then we set $b_e = b_y$ and it holds that
	\begin{align*}
		\sum_{t' \in T_x} \Delta_{t,p,s}^{t'}(e) &= \sum_{t' \in T_x} \Delta_{t,p,s}^{t'}(\offset(y, o))\\
		&= \sum_{t' \in T_x} \Delta_{t,p,s}^{\olast(y, t', p, o)}(y) && \text{Def. $\Delta$}\\
		&= \sum_{t' \in T_x} \Delta_{t,p,s}^{t'}(y) &&\text{$\olast$ is a bijection}\\
		&= \sum_{t' \in T_y} \Delta_{t,p,s}^{t'}(y) &&\text{$p$ well-typed ($T_x \subseteq T_y$)}\\
		&= b_y = b_e &&\text{Induction Hypothesis}
	\end{align*}
	$\olast$ is a bijection because of the assumption that $p$ is correctly typed.

	\vspace{3mm}
	\item \textbf{Window Aggregation (sum):} $e = \mathsf{aggr}(y, W, \Sigma)$.
	Because we know $p$ is well-typed, there exists a frequency $\delta_x$ of the stream containing $e$.
	Then we set $b_e = (\floor{\frac{W}{\delta_x}} + 1) \cdot b_y$ and it holds that:
	\begin{align*}
		\sum_{t' \in T_x} \Delta_{t,p,s}^{t'} (e) &= \sum_{t' \in T_x} \Delta_{t,p,s}^{t'}(\mathsf{aggr}(y, W, \Sigma))\\
		&= \sum_{t' \in T_x} \sum_{t'' \in \windowtimes(y, t', p, W)} \Delta_{t,p,s}^{t''}(y) &&\text{Def. $\Delta$}\\
		&= \sum_{t' \in T_y} \Delta_{t,p,s}^{t'}(y) \cdot m(t')
	\end{align*}
	where $m(t') = |\setCond{t''\in T_x}{t'\in \windowtimes(y, t'', p, W)}|$.
	We claim that
	\[
		\forall t' \in T_y. m(t') \le \floor{\frac{W}{\delta_x}} + 1.
	\]
	and then
	\begin{align*}
		\sum_{t' \in T_y} \Delta_{t,p,s}^{t'}(y) \cdot m(t') &\le \sum_{t' \in T_y} \left(\Delta_{t,p,s}^{t'}(y) \cdot \left(\floor{\frac{W}{\delta_x}} + 1\right)\right) &&\text{Claim}\\
		&= \left(\floor{\frac{W}{\delta_x}} + 1\right) \cdot \sum_{t' \in T_y} \Delta_{t,p}^{t'}(y) &&\text{Def. $\Sigma$}\\
		&\le \left(\floor{\frac{W}{\delta_x}} + 1\right) \cdot b_y  = b_e &&\text{Ind. Hypothesis}
	\end{align*}
	\end{itemize}

	\noindent
	It remains to show that our claim $m(t') \le \floor{\frac{W}{\delta}} + 1$ holds.
	$m(t')$ counts how many values of x at timestamp $t'' \in T_x$ aggregate over a window that contains the value of y at timestamp $t' \in T_y$.

	The aggregations of $x$ that include the value $y$ at timestamp $t'$ are in the real-time window between $p(t)$ and $p(t) + W$.
	Because $p$ is well-typed (and the aggregation $e$ is contained in the stream-expression for $x$), it holds that
	\[
		T_x = \setCond{t' \in \time}{\exists k. p(Y)(t') = k\cdot \delta_x}.
	\]
	Because the timestamps in $T_x$ are spaced by $\delta_x$ in real time, any interval of length $W$ can contain at most 
	\[
	\floor{\frac{W}{\delta_x}} + 1
	\] 
	timestamps of $T_x$.
	Therefore, for every $t' \in T_y$, we have
	\[
	m(t') = |\{ t'' \in T_x \mid t' \in \windowtimes(y, t'', p, W) \}| \le \floor{\frac{W}{\delta_x}} + 1,
	\]
	as required.

	\begin{itemize}
		\item \textbf{Window Aggregation (last):} $e = \mathsf{aggr}(y, W, last)$.
		Then we set $b_e = (\floor{\frac{W}{\delta_x}} + 1) \cdot b_y$ and it holds that:
		\begin{align*}
			\sum_{t' \in T_x} \Delta_{t,p,s}^{t'}(e) &= \sum_{t' \in T_x} \Delta_{t,p,s}^{t'}(\mathsf{aggr}(y, W, last))\\
			&\le \sum_{t' \in T_x} \Delta_{t,p,s}^{t'}(\mathsf{aggr}(y, W, \Sigma)) && \text{$\Sigma$ contains last}\\
			&\le \left(\floor{\frac{W}{\delta_x}} + 1\right) \cdot b_y  = b_e &&\text{Proof for $\Sigma$}
		\end{align*}

		\item \textbf{Bounded hold:} $e = \mathsf{hold}_{bound}(y)$. Then we set $b_e = b_y \cdot bound$ and it holds that:
		\begin{align*}
		\sum_{t' \in T_x} \Delta_{t,p,s}^{t'}(e)
		&= \sum_{t' \in T_x} \Delta_{t,p,s}^{t'}(\mathsf{hold}_{bound}(y))\\
		&\le \sum_{t' \in T_x} \begin{cases}
			\Delta_{t,p,s}^{\olast(y, t', 0, p)}(y) & \text{if } holdn(x, y, t') \le bound\\
			0 & \text{otherwise}
		\end{cases}\\
		&\le \sum_{t' \in T_y} \Delta_{t,p,s}^{t'}(y) \cdot m(t')\\
		&\le b_y \cdot m(t') &&\hspace{-3cm}\text{Induction hypothesis}\\
		&\le b_y \cdot bound &&\hspace{-3cm}\text{$m(t') \le bound$}
		\end{align*}
		where
		\[
		m(t) = |\setCond{t' \in T_x}{\olast(y, t', 0, p) = t \land holdn(x, y, t', p) \le bound}|
		\]
		is bounded by $bound$, because there can be at most $bound$ distinct timepoints $t' \in T_x$ that map to the same $t$ before the hold counter $holdn(x, y, t', p)$ exceeds $bound$.
	\end{itemize}

	\noindent
	For the extension presented in \Cref{sec:value_dependent}, we have shown in \Cref{app:def_influence} that 
	\[
		\exists n_x . \forall p \in \pacingModel_\varphi . \forall t \in \time. |\setCond{t' \in T_e}{t \rightsquigarrow_p^x t'}| \le n_x.
	\]
	Then, 
	\begin{align*}
		\sum_{t' \in T_x} \Delta_{t,p,s}^{t'}(e) &= \sum_{t' \in T'} ub(x) - lb(x)\\
		\intertext{where $T' = \setCond{t' \in T_x}{t \rightsquigarrow_p^x t'}$}
		&= |T'| \cdot (ub(x) - (lb(x)))\\
		&\le n_x \cdot (ub(x) - (lb(x)))
	\end{align*}
\end{proof}

\section{Proof for Differential Privacy}\label{app:proof_dp}

\begin{proof}
Let $\world,\world' \in \World$ be $s$-distant, event-level adjacent evaluation models, differing at timestamp $t$.
By Theorem~\ref{thm:delta}, for every timestamp $t'$,
\[
|\world(x)(t') - \world'(x)(t')| \le \Delta_{t,p,s}^{t'}(x),
\]
where $p$ is the common pacing model.

Summing over all timestamps at which $x$ is evaluated and applying Theorem~\ref{thm:bound} yields
\[
\sum_{t' \in T_{p,x}} |\world(x)(t') - \world'(x)(t')| 
\le \sum_{t' \in T_{p,x}} \Delta_{t,p,s}^{t'}(x) \le b_{x,s}.
\]

Thus, the total $L_1$-sensitivity of the released stream is bounded by $b_{x,s}$.
The result now follows directly from \Cref{thm:laplace}: adding independent Laplace noise with scale $b_{x,s}/\varepsilon$ to each output value ensures $s$-distant, $\varepsilon$-event-level differential privacy.
\end{proof}

\section{Tree-Based Aggregation}\label{app:trees}

\begin{remark}
	For this section we assume that we aggregate over a stream $x$ which has a value at every timepoint.
	If the original stream does not, we introduce an auxiliary stream by deterministically replacing each $\bot$ with the neutral element of the aggregation function.
	This transformation is independent of the data and does not affect differential privacy.
\end{remark}

The nodes of the tree partition the timepoints into intervals.
Let $\mathcal{I}_k$ be a partition of the timepoints into the nodes in level $k$ of the tree:
\[
	\mathcal{I}_k = \set{I_{k,1},I_{k,2},\ldots}.
\]
For each level $k$ of the tree, we define a separate mechanism, which returns all the aggregates for each interval on level $k$ given a stream $x\in \Iref \uplus \Oref$ in an evaluation model $\world \in \World$:
\[
	M_{k,x}(\world) = \left(\sum_{t\in I_{k,j}}\world(x)(t)\right)_{j=1}^\infty
\]
and we define $\tilde{M}_k(x) = M_k(x) + \eta_{k}$ as a perturbed version where $\eta_{k,j} \sim Lap(\frac{b_{x,s}}{\varepsilon_k})$ and $b_{x,s}$ is the sensitivity bound obtained for the stream $x$ through \Cref{thm:bound}.

This mechanism releases all nodes at level $k$ at once.

\begin{theorem}
	Given a valid evaluation model $\world$ the mechanism $\tilde{M}_{k,x}(\world)$ satisfies $s$-distant, $\varepsilon_k$-event-level differential privacy.
\end{theorem}

\begin{proof}
	We compute an overapproximation of the L1-sensitivity of $M_{k,x}$.
	Given two $s$-distant, event-level adjacent evaluation models $\world$ and $\world'$:
	\begin{align*}
	\left\|M_{k,x}(\world) - M_{k,x}(\world')\right\|_1 &= \sum_{j=1}^\infty \left|M_{k,x}(\world)_j - M_{k,x}(\world')_j\right|\\
	&= \sum_{j=1}^\infty \left| \sum_{t\in I_{k,j}}\world(x)(t) - \world'(x)(t)\right|\\
	&\le \sum_{j=1}^\infty \sum_{t\in I_{k,j}} |\world(x)(t) - \world'(x)(t)|\\
	&\le \sum_{t=1}^\infty |\world(x)(t) - \world'(x)(t)| \cdot |\setCond{j}{t \in I_{k,j}}|
	\intertext{
		$I_k$ is a partition of all time points and therefore $|\setCond{j}{t \in I_{k,j}}| = 1$.
	}
	&= \sum_{t=1}^\infty |\world(x)(t) - \world'(x)(t)|\\
	&\le \sum_{t=1}^\infty \Delta_{t',p,s}^{t} \hspace{2cm}\text{\Cref{thm:delta}}
	\intertext{
		where $t'$ is the timepoint where $\world$ and $\world'$ differ.
	}
	&\le b_{x,s} \hspace{3cm}\text{\Cref{thm:bound}}
	\end{align*}
	As we have shown that $b_{x,s}$ is an overapproximation of the sensitivity of $M_{k,x}$, it follows from \Cref{thm:laplace} that $\tilde{M}_{k,x}$ satisfies $\varepsilon_k$-differential privacy w.r.t the event-level, $s$-distance adjacency relation.
\end{proof}

\begin{theorem}
	The mechanism $\tilde{M}(x) = (\tilde{M}_1(x),\tilde{M}_2(x),\ldots)$ satisfies $s$-distant, $\varepsilon$-private event-level differential privacy when $\sum_{k=1}^\infty \varepsilon_k \le \varepsilon$.
\end{theorem}
This follows directly from the composition theorem of differential privacy.
One example for $\varepsilon_k$ is using the geometric series with $\varepsilon_k = \frac{6}{\pi^2 \cdot k^2} \cdot \varepsilon$.

The mechanism releases all perturbed tree nodes.
For each timepoint, the reported aggregate is obtained by combining a finite subset of these nodes whose intervals partition the interval from the start up to this timepoint.
This combination is pure post-processing.

\section{RTLola Extension Constructs}

To build these constructs, we make use of RTLola features not covered in this paper.
For more information, we refer to the RTLola Tutorial~\cite{DBLP:conf/fm/BaumeisterFKS24}.

\subsection{Bounded Hold}

Given a bounded hold access from \lstinline!x! to \lstinline!y!:
\begin{lstlisting}
	output x @px@ := ...
	output y @py@ := x.hold(for_discrete: b).defaults(to: dft)
\end{lstlisting}
we can replace the bounded hold access by utilizing the annotated (or derived) pacing types and introducing a counter stream:
\begin{lstlisting}
	output x @px@ := ...
	output hold_counter
		eval @px@ with 0
		eval @py@ with hold_counter.offset(by:-1).defaults(to: 0) + 1
	output y @py@ := if hold_counter > b then dft else x.hold(or: dft)
\end{lstlisting}

\subsection{Tree Based Aggregation}

As an example, consider the following over-all sum aggregation:
\begin{lstlisting}
input x : UInt64
output y := x.aggregate(over_discrete: all, using: sum)
\end{lstlisting}
we can express this aggregation using the tree-based approach:
\begin{lstlisting}
input x : UInt64
output time : UInt64 @x@ := time.last(or: 0) + 1
output height := log_2(time)
output layers(level)
  spawn with height
  eval @x@ when level == 0 with a
  eval @x@ when (time - 1) % (2**level) == 0 with a
  eval @x@ when (time - 1) % (2**level) != 0 with layers(level).last(or: 0) + x

output private_layer(level)
  spawn with height
  eval @x@ when time % (2**level) == 0 with layers(level).hold(or: 0) + laplace(1/(level+1))

output layer_contribution(level)
  spawn with height
  eval @x@ when ((time >> level) & 1) == 1 with private_layer(level).hold(or: 0)

output y := layer_contribution.aggregate(over_instances: fresh, using: sum)
\end{lstlisting}
A sliding window \lstinline!output b @1s@ := a.aggregate(over: 3s, using: sum)! can be expressed as
\begin{lstlisting}
output time @1s@ := time.last(or: 0) + 1
output time_m @1s@ := time % 6
output l0 @1s@ := a.aggregate(over: 1s, using: sum)
output l0_n @1s@ := l0 + laplace(1/2)
output l1 @1s@ := if (time + 1) % 2 == 0 then l0 + l0.offset(by:-1).defaults(to: 0.0) else 0.0
output l1_n @1s@ := l1 + laplace(1/2)
output b @1s@ :=
  if time_m == 0 || time_m == 2 || time_m == 4 then
    l1_n + l0_n.offset(by: -2).defaults(to: 0.0)
  else if time_m == 1 || time_m == 3 || time_m == 5 then
    l1_n + l0_n.offset(by: -2).defaults(to: 0.0)
  else 0.0
\end{lstlisting}

\section{Specifications}

\subsection{Fraud Detection}
\label{app:fraud_detection}

\begin{lstlisting}
#[sensitivity="unbounded"]
input target: Float64
#[range_from="0", range_to="1000"]
input amount : Float64

output total_amount_per_month @30d := amount.aggregate(over: 30d, using: sum)
output average_amount_per_month @30d := amount.aggregate(over: 30d, using: average).defaults(to: 0.0)

#[public]
output t0 := total_amount_per_month > 50000.0
#[public]
output t1 := average_amount_per_month > 700.0

output amount_per_target(t)
  spawn with target
  eval with if t == target then amount else 0.0

output total_amount_per_target_per_month(t)
  spawn with target
  eval @Global(30d)@ with amount_per_target(t).aggregate(over: 30d, using: sum)

output max_amount_to_single_target_per_month
  eval @30d@ with total_amount_per_target_per_month.aggregate(over_instances: all, using: max).defaults(to: 0.0)

output total_max_ratio @30d@ := max_amount_to_single_target_per_month / total_amount_per_month

#[public]
output t2 := max_amount_to_single_target_per_month > 10000.0
#[public]
output t3 := total_max_ratio > 0.5
\end{lstlisting}

\subsection{Public Transportation}
\label{app:public_transportation}

\begin{lstlisting}
#[sensitivity="unbounded"]
input line : UInt64
#[range_from="1", range_to="10"]
input crowdedness : Float64

output crowdedness_per(l)
  spawn with line
  eval with if line == l then crowdedness else 0.0

output one_if(l)
  spawn with line
  eval with if line == l then 1.0 else 0.0

output sum_crowdedness(l)
  spawn with line
  eval @1h@ with crowdedness_per(l).aggregate(over: 1h, using: sum)

output count_crowdedness(l)
  spawn with line
  eval @1h@ with one_if(l).aggregate(over: 1h, using: sum)

output avg_crowdedness(l)
  spawn with line
  eval @1h@ with sum_crowdedness(l) / count_crowdedness(l)

output avg_crowdedness_clipped(l)
  spawn with line
  eval @1h@ with if avg_crowdedness(l) > 10.0 then 10.0 else if avg_crowdedness(l) < 0.0 then 0.0 else avg_crowdedness(l)

#[public]
output avg_124
  eval @1h@ when count_crowdedness(1.0).hold(or: 0.0) > 5.0 with avg_crowdedness_clipped(1.0).hold(or: 0.0)
#[public]
output avg_101
  eval @1h@ when count_crowdedness(2.0).hold(or: 0.0) > 5.0 with avg_crowdedness_clipped(2.0).hold(or: 0.0)
#[public]
output avg_party
  eval @1h@ when count_crowdedness(3.0).hold(or: 0.0) > 5.0 with avg_crowdedness_clipped(3.0).hold(or: 0.0)
\end{lstlisting}

}{}

\end{document}